\documentclass[journal,letterpaper,twocolumn,twoside,nofonttune]{IEEEtran}
\usepackage{etex}
\normalsize
\usepackage[T1]{fontenc}
\usepackage{amsmath,amssymb,amsfonts}
\usepackage{mathrsfs}
\usepackage{mathabx}
\usepackage{amsbsy}
\usepackage{graphicx}
\usepackage{graphics}
 \usepackage{epstopdf}
\usepackage{algorithm}
\usepackage{algorithmic}
\usepackage{subfigure}
\usepackage{cite}
\usepackage{multirow}
\usepackage{ctable}
\usepackage{caption}
\usepackage{setspace}
\usepackage{mathtools}
\usepackage{lipsum}
\usepackage[algo2e]{algorithm2e}

\title{\Huge$\,$\\[-2.75ex]
{Distributed Multi-User Secret Sharing}\\[0.50ex]}

\author{\large%
Mahdi Soleymani 
and 
Hessam Mahdavifar,\,\,\IEEEmembership{Member,~IEEE}\\

\thanks{%
 The material in this paper was presented in part at the IEEE International Symposium on Information Theory in June 2018.
}
\thanks{This work was supported by the National Science Foundation
under grants CCF--1763348, CCF--1909771, and CCF--1941633.}
\thanks{M.\ Soleymani and H.\ Mahdavifar are with the Department of Electrical Engineering and Computer Science, University of Michigan, Ann Arbor, MI 48104 (email: mahdy@umich.edu and hessam@umich.edu).}
\thanks{Copyright (c) 2017 IEEE. Personal use of this material is permitted.  However, permission to use this material for any other purposes must be obtained from the IEEE by sending a request to pubs-permissions@ieee.org.
}
}

\newtheorem{exmp}{Example}[section]
\newtheorem{theorem}{{Theorem}}
\newtheorem{lemma}[theorem]{{Lemma}}

\newtheorem{corollary}[theorem]{{Corollary}}

\newtheorem{definition}{{Definition}}

\newcommand{\cD}{{\cal D}}
\newcommand{\cE}{{\cal E}} 
\newcommand{\cF}{{\cal F}}

\newcommand{\cJ}{{\cal J}}

\DeclareMathAlphabet{\mathbfsl}{OT1}{ppl}{b}{it} 

\newcommand{\bA}{\mathbfsl{A}}

\newcommand{\by}{\mathbfsl{y}} 

\newcommand{\bs}{\mathbfsl{s}}

\newcommand{\ceil}[1]{\left\lceil #1 \right\rceil}
\newcommand{\floor}[1]{\left\lfloor #1 \right\rfloor}

\newcommand*{\rom}[1]{\expandafter\romannumeral #1}

\makeatletter
\newcommand{\AlignFootnote}[1]{%
  \ifmeasuring@
  \else
    \iffirstchoice@
      \footnote{#1}%
    \fi
  \fi}
\makeatother

\newcommand{\be}[1]{\begin{equation}\label{#1}}
\newcommand{\ee}{\end{equation}} 
\newcommand{\eq}[1]{(\ref{#1})}

\renewcommand{\leq}{\leqslant}
 
\renewcommand{\geq}{\geqslant}

\newcommand{\script}[1]{{\mathscr #1}}

\renewcommand{\Bbb}{\mathbb}
 
\newcommand{\N}{{\Bbb N}}

\newcommand{\F}{{\Bbb F}}

\newcommand{\Tref}[1]{Theo\-rem\,\ref{#1}}

\newcommand{\Lref}[1]{Lem\-ma\,\ref{#1}}
\newcommand{\Cref}[1]{Co\-ro\-lla\-ry\,\ref{#1}}

\newcommand{\Fq}{{{\Bbb F}}_{\!q}}

\newcommand{\deff}{\mbox{$\stackrel{\rm def}{=}$}}

\renewcommand{\labelenumi}{\theenumi}

\newcommand{\sA}{\script{A}}

\begin{document}

\maketitle

\begin{abstract}
We consider a distributed secret sharing system that consists of a dealer, $n$ storage nodes, and $m$ users. Each user is given access to a certain subset of storage nodes, where it can download the stored data. The dealer wants to securely convey a specific secret $s_j$ to user $j$ via storage nodes, for $j=1,2,\dots,m$. More specifically, two secrecy conditions are considered in this multi-user context. The \textit{weak} secrecy condition is that each user does not get any information about the individual secrets of other users, while the \textit{perfect} secrecy condition implies that a user does not get any information about the collection of all other users' secrets.
In this system, the dealer encodes secrets into several secret shares and loads them into the storage nodes. Given a certain number of storage nodes we find the maximum number of users that can be served in such a system and construct schemes that achieve this with perfect secrecy. We further define two major properties for such distributed secret sharing systems; \emph{communication complexity} is defined as the total amount of data that users need to download in order to reconstruct their secrets; and \emph{storage overhead} is defined as the total size of data loaded by the dealer into the storage nodes normalized by the total size of secrets. Lower bounds on the minimum communication complexity and the storage overhead are characterized given any $n$ and $m$. We construct distributed secret sharing protocols, under certain conditions on the system parameters, that attain the lower bound on the communication complexity while providing perfect secrecy. Furthermore, we construct protocols, again under certain conditions, that simultaneously attain the lower bounds on the communication complexity and the storage overhead while providing weak secrecy, thereby demonstrating schemes that are optimal in terms of both parameters. It is shown how to modify the proposed protocols in order to construct schemes with \emph{balanced} storage load and communication complexity.
\end{abstract}

\begin{IEEEkeywords}
Secret sharing, distributed storage, multi-user security
\end{IEEEkeywords}

\section{Introduction} 
\label{sec:Introduction}

\IEEEPARstart{S}{ecret} sharing, introduced by Shamir \cite{shamir1979share} and Blakely \cite{blakley1979safeguarding}, is central in many cryptographic systems. It has found applications in cryptography and secure distributed computing including secure interactive computations \cite{ben1988completeness,chaum1988multiparty,beaver1991foundations,canetti2000security,cramer2000general}, secure storage \cite{garay2000secure,ateniese2006improved,kumar2012publicly}, generalized oblivious transfer \cite{shankar2008alternative,tassa2011generalized}, threshold cryptography \cite{desmedt1992shared,desmedt1993threshold,shoup2000practical}, and secure matrix multiplication \cite{aliasgari2020private,aliasgari2019distributed}. A secret-sharing scheme involves a dealer, who has a secret, a set of users, and a collection $A$ of subsets of users, which is called the access structure. A secret-sharing scheme for the access structure $A$ is a scheme for distributing the secret by the dealer among the users while guaranteeing the following. 1) Secret recovery: any subset in the access structure $A$ can recover the secret from its shares; 2) Collusion resistance: for any subset not in $A$, the aggregate data of users in that subset reveals no information about the secret. 

Most cryptographic protocols involving secret sharing assume that the central user, called the dealer, has a direct reliable and secure communication channel to all the users. In such settings, it is assumed that once the dealer computes the shares of secret, they are readily available to the users. In many scenarios, however, the dealer and users are nodes of a large network. In general, the communication between the dealer node and users can be through several relay nodes, as in a relay network or through intermediate network nodes, as in a network coding scenario. Alternatively, in a distributed storage scenario, the dealer can be thought of as a master node controlling a certain set of servers or storage nodes, while each user has access to a certain subset of servers. 

In this paper, we consider the later scenario. In particular, the system model is shown in Figure\,\ref{fig}. The dealer is considered as a central entity that controls a given set of servers, also referred to as storage nodes, and can load data to them. Alternatively, in an application concerning multiple-access wireless networks, one can think of middle nodes, sitting between the dealer and the users, as resource elements in different time or frequency, while each user has access to a certain subset of resource elements. We further consider a multi-user secret sharing scenario, in the sense that there is a designated secret, independently generated for each user, to be conveyed to that user and construct protocols to this end. The secret sharing protocols proposed in Section\,\ref{sec:three} and Section\,\ref{sec:four} are \emph{perfectly} secure,  and the protocols provided in Section\,\ref{sec:five} are \emph{weakly} secure. The system model, the weak and the perfect secrecy conditions, and our approach to construct secret sharing protocols for this system are described next. 

\subsection{System model}
\label{sec:two}
A distributed secret sharing system, shown in Figure\,\ref{fig}, consists of a dealer, $n$ storage nodes, and $m$ users. The goal of this system is to enable the dealer to securely convey a specific secret to each user via storage nodes. In this system model:
 \renewcommand{\labelenumi}{\alph{enumi})}
  \begin{enumerate}
  
\item For each user $j$, $A_j \subseteq [n]$, where $[n]\ \deff \ \{1,2,3,...,n\}$, denote the set of all storage nodes that user $j$ has access to. The set $A_j$ is referred to as the \emph{access set} for the user $j$. For each $i \in A_j$, user $j$ can read the entire data loaded into node $i$. Let 
\be{def-A}
\script{A} \,\deff\, \{ A_j : j \in [m] \}
\ee
denote the set of all access sets, which is called the \emph{access structure}.

\item Storage nodes are \emph{passive}; they do not communicate with each other. Also, the users do not communicate with each other.

\item Let $s_j \in \mathbb {F}_q$ denote the secret for user $j$. Also, $s_j$'s are uniformly distributed and mutually independent.
\item The dealer has access to all the storage nodes but it does not have direct access to the users. 
\end{enumerate}


\begin{figure}
\centering
\includegraphics[width=\linewidth]{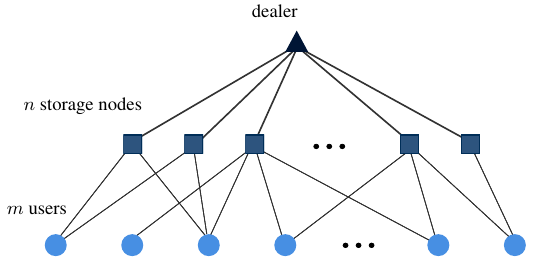}
\caption{System model}
\label{fig}
\end{figure} 

 We aim at designing distributed secret sharing protocols to encode the secrets into secret shares, and distribute them in the storage nodes in such a way that: 1) Each user $j$ can successfully reconstruct its designated secret $s_j$; 2) The secrecy condition, in a \textit{weak} or a \textit{perfect} sense, is satisfied. This is defined more precisely as follows.
 \begin{definition}\label{ssd}
A \emph{distributed secret sharing protocol} (DSSP) is a bundle of $(\script{A},\cE,\bold{Z}_{n\times h},\bold{\cD})$, where
 \renewcommand{\labelenumi}{\roman{enumi})}
  \begin{enumerate}
\item $ \script {A} $ is the access structure as defined in \eq{def-A}.
\item $\cE: \mathbb{F}_q^m \rightarrow \mathbb{F}_q^h $ is an encoding function, for some $h \geq m$, which relates to the \emph{storage overhead} of the system as specified later. The input to the encoding function $\cE$ is $\bold{s}=(s_1,s_2,...,s_m)^T$ as the vector of all secrets. The output $\bold{y} = \cE(\bold {s})$ is the vector of all data to be distributed and stored in the storage nodes. Let $\bold{y} = (y_1,y_2,\dots,y_h)^T$.
\item $\bold{Z} = [z_{i,r}]_{n \times h}$, where $z_{i,r}=1$ if $y_r$ is stored in $i$-th storage node, otherwise $z_{i,r}=0$. The matrix $\bold{Z}$ is referred to as the \emph{storing matrix} and the outcome of this, specifying which data symbols are stored at each storage node, is referred to as the \textit{storage profile}. Let $\bold{y}_j$ denotes the vector of all data stored in nodes indexed by elements of the access set $A_j$. 
\item $\cD$ is a collection of $m$ decoding functions $\cD_j: \mathbb{F}_q^{|\bold{y}_j|}\rightarrow \mathbb{F}_q $, for $j\in [m]$, such that $\cD_j(\bold{y}_j)=s_j$. In other words, each user is able to successfully reconstruct its own secret. This is referred to as the \emph{correctness condition}.
\end{enumerate}
Also, in the protocol, the \emph{weak secrecy} condition is that a user $j$ does not get any information, in an information-theoretic sense, about the individual secret of another user $l$, i.e., 
\begin{equation}\label{sec-con}
\forall j,l\in [m],l \neq j : \quad H(s_l | \bold{y}_j)=H(s_l),
\end{equation}
\end{definition}
and the \emph{perfect secrecy} condition is that no information about the collection of all other secrets is revealed to the user $j$, i.e.,
\be{per-sec-con}
\forall j\in [m]: \quad H(\bs_{-j}|\by_j)=H(\bs_{-j}).
\ee
A DSSP satisfying the weak secrecy condition is also referred to as a weakly secure DSSP and a DSSP satisfying the perfect secrecy condition is also referred to as a perfectly secure DSSP. It is clear that a perfectly secure DSSP is also a weakly secure DSSP.

Note that one can always process the data using source coding techniques prior to the encoding step in the master node. Therefore, it is natural to assume that $s_j$'s are independent and uniformly distributed. The notion of \emph{weak secrecy}, as described by \eqref{sec-con}, is also used in the context of network coding \cite{bhattad2005weakly} as well as distributed storage \cite{kadhe2014weakly}.
Note also that in a protocol with perfect secrecy, user $j$ does not obtain any information about the \emph{vector} of all other users' secret $\bs_{-j}$, while, in a one with weak secrecy, it might infer some information about $\bs_{-j}$, e.g., a linear combination of $s_1,\cdots,s_{j-1},s_{j+1}, \cdots, s_m$ might be revealed to user $j$. However, as pointed out in \cite{kadhe2014weakly} in the context of distributed storage, the perfect secrecy is too strong for certain practical distributed storage systems and may result in extremely large \textit{storage overhead}. We also observe a similar pattern in our setup. The proposed DSSPs in Section\,\ref{sec:three} and Section\,\ref{sec:four} satisfy the perfect secrecy condition but result in rather large \textit{storage overheads}. As a result, we also consider the relaxed notion of weak secrecy in order to design DSSPs with optimal \textit{storage overhead}. 

The notions of \emph{storage overhead} as well as \emph{communication complexity} are defined next. These are used throughout the paper in order to evaluate the efficiency of the proposed DSSPs.
 
Note that the total number of $\Fq$-symbols stored in the storage nodes is 
 \begin{align*}
  & k'\,=\,\sum_{i=1}^{n}\sum_{r=1}^{h} z_{i,r},
  \end{align*}
where $\bold{Z} = [z_{i,r}]_{n \times h}$ is specified in Definition\,\ref{ssd}. The storage overhead, also referred to as SO, of the DSSP is then defined as
\begin{equation}\label{so}
 \text{SO}\ \deff \ \frac{ k'}{m}.
 \end{equation}
Note that the correctness condition must be satisfied for $m$ mutually independent and uniformly distributed secrets. Therefore, $k' \geq m$ and, consequently, $\text{SO} \geq 1$.
This lower bound is not in general tight under the perfect secrecy condition (see \Lref{lem_so} in Appendix B). A similar result, i.e., the lower bound not being tight under perfect secrecy, is shown in a related work studying the single user case \cite{huang2016communication}. Obtaining a tight lower bound for the SO under perfect secrecy is left for future work and is not considered in this paper. However, we show that the lower bound SO $=1$ can be achieved under the weak secrecy condition for a certain set of parameters, thereby providing a protocol with weak secrecy and the optimal SO equal to $1$. 

 Let $c_j$ denote the total number of symbols that the user $j$ needs to download from the storage nodes in the access set $A_j$ in order to reconstruct $s_j$. Note that $c_j \leq |\bold{y}_j|$, since user $j$ may not need to download all its accessible data. Then the \emph{communication complexity} $C$ is defined as
 \begin{equation}\label{c}
 C\ \deff\ \sum_{j=1}^{m} c_j.
 \end{equation}
  
\subsection{Main Results}

We first consider the problem of finding the maximum number of users that can be served in a DSSP given a certain number of storage nodes. This maximum number is derived using a necessary and sufficient condition on access sets in a DSSP that relates to Sperner families in combinatorics. We further present a method for constructing perfectly secure DSSPs that serve the maximum number of users.

For a given number of users $m$ and number of storage nodes $n$, a DSSP with minimum communication complexity $C$ defined in \eq{c} is called a \emph{communication-optimal} DSSP.
We solve a discrete optimization problem to provide a lower bound on the minimum communication complexity. We further construct perfectly secure DSSPs that are communication-optimal when $m$ is a binomial coefficient of $n$, i.e., they achieve the minimum possible communication complexity while providing perfect secrecy. 

We further construct weakly secure DSSPs with nearly optimal and optimal storage overhead for any given parameters $m$ and $n$. In the proposed scheme with optimal storage overhead, no external randomness is required and the total size of data to be stored in storage nodes is equal to the total size of secrets. Consequently, this provides the optimal storage overhead. Combining this with communication-optimal DSSPs yields optimality of both the communication complexity and the storage overhead with weak secrecy. 

Finally, it is shown how to modify the optimal constructions of DSSPs in order to have both balanced storage load, i.e., the number of data symbols stored in each node and balanced communication complexity, i.e., the number of symbols downloaded from each node.  

 
 
\subsection{Shamir's Scheme and Related Works}
 The $(k,t)$ secret sharing scheme proposed by Shamir in \cite{shamir1979share} is described next. Given a secret $s \in \mathbb{F}_q$ the output of the scheme consists of $k$ \emph{secret shares} $d_1,d_2,\dots,d_k \in \mathbb{F}_q$ in such a way that:
 
 \begin{enumerate}
 \item The secret $s$ can be reconstructed given any $t$ or more of the secret shares;

 \item The knowledge of any $t-1$ or fewer secret shares does not reveal any information about $s$, in the information-theoretic sense.
 \end{enumerate}
 To this end, a $(t-1)$-degree polynomial $P(X)$ is constructed as
\be{shamir polynomial}
P(x)=s+\sum_{i=1}^{t-1} p_i x^i,
\ee
where $p_i$'s are i.i.d and are selected uniformly at random from $\mathbb{F}_q$. Let $\gamma_1,\gamma_2,\dots,\gamma_k$ denote $k$ distinct non-zero elements from $\Fq$. The secret shares are then constructed by evaluating $P(x)$ at $\gamma_i$'s, i.e., 
$$
\forall i \in [k] \quad d_i=P(\gamma_i).
$$
We refer to the encoder $\cE:\Fq \rightarrow \Fq^k$ that takes $s$ as the input and outputs $(d_1,d_2,\dots,d_k)$ as described above as a $(k,t)$ Shamir's secret encoder. 

Given any $t$ secret shares $P(x)$ is interpolated and is uniquely determined. This is because the degree of $P(x)$ is at most $t-1$ and each secret share specifies one interpolation point leading to $t$ distinct interpolation points. Then $s=P(0)$ is reconstructed. We refer to this process as Shamir's secret decoder. The process of generating secret shares by evaluating $P(X)$ and then recovering the secret by interpolating it is essentially same as Reed-Solomon encoding and decoding process. The close connection between the two was noted in \cite{mceliece1981sharing}.

There are several previous works that have considered Shamir's scheme in the context of networks \cite{shah2015distributed} and distributed storage systems \cite{bitar2017staircase,huang2016communication}. In these works, there is only one secret, as in the original Shamir's scheme, to be distributed to users either as nodes of a network \cite{shah2015distributed} or as users of a distributed storage system \cite{bitar2017staircase,huang2016communication}, in a collusion-resistant way. However, we consider a multi-user secret sharing scenario, where there is one designated secret for each user, and the secret shares are distributed over a set of storage nodes. Also, in our constructed schemes we guarantee that each user does not get any information about other users' secrets, either individually in a weak sense or collectively in a perfect sense, thereby providing proper measures of security in a multi-user setting. 

\section{DSSP with Maximum Number of Users}
\label{sec:three}

In this section, we consider the following problem: What is the maximum possible number of users that can be served in a DSSP given a certain number of storage nodes? A necessary condition on access sets in a DSSP with weak secrecy is shown which relates to Sperner families in combinatorics. Furthermore, it is shown that this condition is sufficient to guarantee perfect secrecy in a DSSP. In other words, the provided condition is necessary and sufficient for the existence of both weakly secure and perfectly secure DSSPs. In fact, the relation to Sperner families is invoked to present a method for constructing DSSPs that serve maximum number of users.

\begin{lemma}\label{spe-con}
For a weakly secure DSSP with access structure $\script{A}$ defined in \eq{def-A}:
\begin{equation} \label{spe-con-1}
 A_j\nsubseteq A_l,
 \end{equation}
 for all $j,l \in [m]$ with $j \neq l$.
  \end{lemma}
 \begin{proof}: Assume to the contrary that $A_j \subseteq A_l$ for some $j \neq l$. Therefore, the entire accessible data by user $j$ can also be accessed by user $l$. Since user $j$ can retrieve $s_j$ by the correctness condition, user $l$ can
 also retrieve $s_j$. This means the weak secrecy condition is violated and the protocol is not a DSSP as defined in Definition 1, which is a contradiction.
\end{proof}

Collections of subsets satisfying the condition of \Lref{spe-con} are well-studied combinatorial objects. Such a collection $\script{A}$ is called a \emph{Sperner family}\cite{sperner1928satz}. For any Sperner family $\script{A}$ we have \cite{sperner1928satz}
\begin{equation}\label{spe-fea}
|\script{A}| \leq { {n}\choose{\floor{n/2}}},
\end{equation}
and more generally a necessary condition for existence of a Sperner family with $a_k$ subsets of size $k$, for $k \in [n]$, is that \cite{lubell1966short} 
\begin{equation}\label{lym}
\sum_{k=0}^{k=n}  \frac{a_k}{{{n}\choose{k}}} \leq 1.
\end{equation}
Since the number of users is $m = |\script{A}|$ , \eq{spe-fea} implies an upper bound on $m$, i.e.,
\begin{equation}\label{uperbound}
m \leq { {n}\choose{\floor{n/2}}}.
\end{equation}

Next, we use Shamir's secret sharing scheme to construct a DSSP with perfect secrecy when the access structure $\sA$ is a Sperner family. Let $t_j = |A_j|$. In this construction, a $(t_j,t_j)$ Shamir's secret sharing scheme is used independently for each user $j$, both in the encoding of $s_j$ by the dealer and decoding it by user $j$. Such a DSSP is denoted by S-DSSP\{$\script{A},n$\}. In other words, the condition in \Lref{spe-con} is a sufficient condition for existence of a perfectly secure DSSP. More specifically, S-DSSP\{$\script{A},n$\} is described as follows:

\renewcommand{\labelenumi}{\roman{enumi})}
\begin{enumerate}
\item $\script{A}$ is a Sperner family consisting of subsets of $[n]$.

\item $\cE(\bold{s}) = (\cE_1(s_1),\cE_2(s_2),\dots,\cE_m(s_m))$, where $\cE_j$ is a $(t_j,t_j)$ Shamir's secret encoder and $t_j = |A_j|$.  

\item $\forall j \in [m] :\ \bold{Z}[a_{i,j},\tau_{j-1}+i] =1$ for $i \in [t_j]$, where $\tau_j =t_1 +...+t_j$ and $\tau_0=0$; and $A_j = \{a_{1,j},a_{2,j},\dots,a_{t_j,j}\}$. All other entries of $\bold{Z}$ are zero. 

\item $\cD_j$ is the $(t_j,t_j)$ Shamir's secret decoder, for $j \in [m]$.  
\end{enumerate}

\begin{lemma}\label{maxuser} S-DSSP\{$\script{A},n$\} is a perfectly secure DSSP satisfying all properties in Definition 1.
\end{lemma}

\begin{proof} S-DSSP\{$\script{A},n$\} assigns a $t_j$-subset of $[n]$ to user $j$. It encodes $s_j$ into $t_j$ secret shares using Shamir's scheme with the threshold $t_j$ and random seeds generated independently from other users. It then stores one share on each node in $A_j$ as specified by $\bold{Z}$. Clearly, each user can reconstruct its secret by invoking Shamir's secret decoder. Also, Shamir's scheme guarantees the perfect secrecy condition, specified in \eq{per-sec-con}, since $\script{A}$ is a Sperner family and consequently, no user other than user $j$ has access to all of its $t_j$ shares. Therefore, S-DSSP\{$\script{A},n$\} is a DSSP satisfying all properties in Definition 1.
\end{proof}

We can pick a Sperner family $\sA$ with the maximum size $|\sA| = {n\choose{\floor{n/2}}}$ and then construct a S-DSSP\{$\script{A},n$\}. This satisfies all properties of a perfectly secure DSSP by \Lref{maxuser} and serves the maximum possible number of users given a certain number of storage nodes $n$. 

\section{Communication-Optimal DSSP}
\label{sec:four}
 
In this section, we derive a lower bound on the communication complexity of DSSPs, i.e., the amount of data that users need to download in order to reconstruct the secrets. We then show \emph{communication-optimal} DSSPs that attain this lower bound under certain conditions. 

Let $m$ and $n$ denote the number of users and storage nodes, respectively. It is assumed that
$$
m \leq {n\choose \floor{n/2}}, 
$$
as in \eq{uperbound}. Otherwise, by \Lref{spe-con} and \eq{spe-fea} a DSSP does not exist.

A DSSP with minimum communication complexity $C$, defined in \eq{c}, is called a \emph{communication-optimal} DSSP. A certain class of DSSPs, called \emph{tight} DSSPs, defined below, is useful to derive lower bounds on the communication complexity and to construct communication-optimal DSSPs. This will be shown in \Lref{allone}.


\begin{definition}\label{tdss}
 A DSSP is said to be a \emph{tight} DSSP (T-DSSP) if every user downloads exactly one $\Fq$-symbol from each node in its access set.  
\end{definition}

Note that, for example, every S-DSSP, defined in Section\,\ref{sec:three}, is a T-DSSP. 


\begin{lemma}\label{allone}
For any DSSP with communication complexity $C$, there exists a perfectly secure T-DSSP with the same number of users and storage nodes, and communication complexity $\tilde{C}$ such that 
\begin{equation*}
\tilde{C} \leq C.
\end{equation*}
\end{lemma}
\begin{proof} For each user $l$, let $\tilde{A_l} \subset A_l$ denote the set with the minimum size such that user $l$ can reconstruct its secret $s_l$ by downloading data from $\tilde{A_l}$. Note that user $l$ has to download at least one symbol from each node in $\tilde{A_l}$. Therefore,
\begin{equation}\label{cbigsA}
\sum_{j=1}^{m} | \tilde{A}_j| \leq C.
\end{equation}
The secrecy condition implies that
\begin{equation}\label{notin}
\forall j,l \in [m], j\neq l: \quad \tilde{A_l} \nsubseteq A_j. 
\end{equation}
Otherwise, user $j$ would be able to reconstruct $s_l$. Let $\tilde{\script{A}}\ \deff \ \{ \tilde{A_j}: \forall j\in [m]\}$, which is a Sperner family. 

We then construct a S-DSSP associated with the access structure $\tilde{\script{A}}$ and having communication complexity $\tilde{C}= \sum_{j=1}^{m} | \tilde{A}_j|$. This follows from the fact that in a S-DSSP each user downloads exactly one data symbol from the nodes in its access set. This together with \eq{cbigsA} and recalling that a S-DSSP is also a T-DSSP complete the proof.
\end{proof}


Note that the communication complexity of a T-DSSP depends only on its associated access structure. Let $a_k$ denote the number of subsets of size $k$ in the access structure of the T-DSSP. Then its communication complexity is given by
\be{tdss-same-c}
C=\sum_{k=1}^{n} k a_k.
\ee
Therefore, by \Lref{allone}, one can consider minimizing $\sum_{k=1}^{n} k a_k$ to find a communication-optimal DSSP provided that a Sperner family with such $a_k$'s exists. To this end, we consider the following discrete optimization problem:
\begin{align}
& \underset{}{\text{min}}
& & \sum_{k=1}^{n} k a_k  \label{min}\\
& \text{s.t.} & &  \forall\  k \in \{1,...,n\}: \ a_k \in \N \cup \{0\} \label{w} \\
& & &  \sum_{k=1}^{n} a_k=m, \label{sum} \\
& & & \sum_{k=1}^{n}  \frac{a_k}{{{n}\choose{k}}} \leq 1. \label{sperner} 
\end{align}
Constraint (\ref{w}) is set because $a_k$'s must be non-negative. Constraint (\ref{sum}) is set because the sum of $a_k$'s is equal to the total number of users $m$. Also, by \eq{lym}, \eq{sperner} is a necessary condition for existence of a Sperner family with $a_k$ subsets of size $k$. If such Sperner family exists for the solution of this optimization problem, then we will have a communication-optimal DSSP with perfect secrecy. Otherwise, the minimum objective function is a lower bound for the minimum communication complexity. Note that due to the reciprocity of the binomial coefficients, we have $a_k=0$ for all $\floor{\frac{n}{2}}<k$ in the solution of this optimization problem.

The idea is to first solve a continuous version of this optimization problem, stated below, and then extract the solution for the discrete version from the solution of the continuous version. Consider the following problem:
 \begin{align}
& \underset{}{\text{min}}
& & \sum_{k=1}^{\floor{n/2}} k \alpha_k  \label{minofc}\\
& \text{s.t.} & &  \forall\  k \in \{1,...,\floor{n/2}\} : \ \  \alpha_k \geq 0  ,\label{wc}  \\
& & &  \sum_{k=1}^{\floor{n/2}} \alpha_k=m, \label{sumc} \\
& & & \sum_{k=1}^{\floor{n/2}}  \frac{\alpha_k}{{{n}\choose{k}}} \leq 1. \label{spernerc} 
\end{align}
where $\alpha_k \in \mathbb{R}$. This optimization problem can be solved by satisfying Karush--Kuhn--Tucker (KKT) condition \cite{boyd2004convex}. Let $\psi^*$ denote the minimum of the objective function in the above continuous optimization problem. Suppose that $\psi^*$ is achieved by the choice of $\alpha_k^*$, for $k=1,2,\dots,\floor{n/2}$. It is shown in Appendix A that at most two of $\alpha_k^*$'s are non-zero. Furthermore, it is shown that if two non-zero $\alpha_k^*$'s exist, then their indices are consecutive. In particular, the solution is described as follows. Let $i$ denote the largest integer such that 
\be{i-def}
 {n\choose i} \leq m.
\ee
Then
\begin{align}\label{alphai2}
&\alpha_i^* =\frac{{n\choose i+1} -m}{{n\choose i+1}-{n\choose i}} {n\choose i},
\\
&\alpha_{i+1}^* =\frac{ m-{n\choose i}}{{n\choose i+1}-{n\choose i}} {n\choose i+1}\label{alphai12},
\end{align}
and $\alpha_k^* = 0$, for $k\neq i, i+1$. Also, by \eq{min_obj} the minimum possible objective function is
\begin{equation}\label{psi}
\psi^* = i \alpha_i ^*+ (i+1) \alpha_{i+1}^*.
\end{equation}


Let $C^*$ denote the minimum of the objective function in the discrete optimization problem. It is clear that
\begin{equation}\label{clowerbound}
\ceil{\psi^*} \leq C^*.
\end{equation}
The following lemma shows that $C^* = \ceil{\psi^*}$. 

\begin{lemma}\label{minc}
We have
$$
C^* = \ceil{\psi^*}.
$$
Furthermore, this minimum objective function is achieved by choosing $a_i=\floor{\alpha_i^*}$, $a_{i+1}=\ceil{\alpha_{i+1}^*}$, and $a_k = 0$, for $k \neq i,i+1$.
\end{lemma}

\begin{proof}
Let $a_i=\floor{\alpha_i^*}$, $a_{i+1}=\ceil{\alpha_{i+1}^*}$, and $a_k = 0$, for $k \neq i,i+1$. Let $\epsilon = \alpha_i^*-\floor{\alpha_i^*}$. Note that $\alpha_i^*+ \alpha_{i+1}^*=m$ and $m$ is an integer. Therefore, $\alpha_{i+1}^*=\ceil{\alpha_{i+1}^*}-\epsilon$. In other words, the fraction part of $\alpha_{i+1}^*$ is $1-\epsilon$. Then one can write
\begin{equation}
\label{clowerbound_ai}
a_i=\alpha_i^*-\epsilon, \quad a_{i+1}=\alpha_{i+1}^*+\epsilon
\end{equation}
where $0\leq \epsilon<1$. First, we show feasibility of this solution by checking the constraints of the optimization problem. It is easy to see that (\ref{w}) and (\ref{sum}) are satisfied. Also,
\begin{align*}
\frac{a_i}{{n\choose i}}+\frac{a_{i+1}}{{n\choose i+1}} = \frac{\alpha_i^*}{{n\choose i}}+\frac{\alpha_{i+1}^*}{{n\choose i+1}}+\epsilon(\frac{1}{{n\choose i+1}}-\frac{1}{{n \choose i}}) \leq 1,
\end{align*}
where the equality holds by \eq{clowerbound_ai} and the inequality holds by (\ref{spernerc}) and noting that $\frac{1}{{n\choose i+1}}-\frac{1}{{n \choose i}}$ is negative . Therefore, (\ref{sperner}) is also satisfied which shows that the solution is feasible. At last, it is shown that this solution achieves equality in \eq{clowerbound}. For this solution, we have 
\begin{align}\label{Cstar}
C^* = ia_i+(i+1)a_{i+1} = i \alpha_i^* + (i+1) \alpha_{i+1}^* + \epsilon = \psi^* + \epsilon.
\end{align}
and, therefore, $ C^*=\ceil{\psi^*}$. 
\end{proof}

The Following theorem is the summary of this section's results.
\begin{theorem}\label{sec:four_main}
For a given number of users $m$ and storage nodes $n$, any T-DSSP with the following access structure $\script{A}$ is a communication-optimal DSSP: $\script{A}$ is a Sperner family that contains $\floor{\alpha_i^*}$ of $i$-subsets of $[n]$ and $\ceil{\alpha_{i+1}^*}$ of $(i+1)$-subsets of $[n]$, where $i$ is the maximum integer that satisfies \eq{i-def}, and $\alpha_{i}^*$ and $\alpha_{i+1}^*$ are as calculated in (\ref{alphai2}) and (\ref{alphai12}).
\end{theorem}
\begin{proof}
The theorem follows by \eq{tdss-same-c} and the solution to the discrete optimization problem with properties shown in \Lref{minc}. 
\end{proof}

\begin{corollary}\label{comeff}
If a Sperner family $\sA$ as specified in \Tref{sec:four_main} exists, then S-DSSP$\{ \script{A}, n \}$ is a communication-optimal and perfectly secure DSSP. Otherwise, $C^*$, given in \eq{Cstar}, is a lower bound on the minimum possible communication complexity.

\end{corollary}

 In particular, if $m$ is a binomial coefficient of $n$, i.e., $m = {n \choose i}$, then a Sperner family $\sA$ exists; $\sA$ is the set of all $i$-subsets of $[n]$. Then S-DSSP$\{ \script{A}, n \}$ is a communication-optimal and perfectly secure DSSP.

It is shown in \Lref{maxuser} that an S-DSSP satisfies the perfect secrecy condition, specified in \eqref{per-sec-con}, since all $s_j$'s are encoded independently by utilizing Shamir's encoder. In fact, the protocols serving the maximum number of the users in Section\,\ref{sec:three} and the communication-optimal DSSPs in this section are S-DSSPs and satisfy the perfect secrecy condition. However, the perfect secrecy comes at the expense of having a large storage overhead, e.g., $SO= \floor{n/2}$ in the S-DSSP serving the maximum number of users. By relaxing the perfect secrecy condition to a weak one, as specified in \eqref{sec-con}, it is possible to significantly reduce the storage overhead of a DSSP as shown throughout the rest of this paper. In particular, it is shown in Section\,\ref{sec:five} that there exists a DSSP with a storage overhead that is optimal and is equal to $1$ while providing weak secrecy, as specified in \eqref{sec-con}.
\section{ DSSPs with Nearly Optimal and Optimal Storage Overhead}
\label{sec:five}

In this section, two schemes are proposed towards constructing DSSPs with optimal storage overhead while providing weak secrecy. In the first scheme, a few number of random symbols are used as external \textit{random seed} for the code construction. Consequently, the storage overhead is slightly greater than one which is shown to be the minimum achievable storage overhead. Hence, the first scheme is referred to as the DSSP with \emph{nearly optimal} storage overhead. Then, the first scheme is modified in order to construct DSSPs with storage overhead equal to one which is optimal. However, the DSSP with nearly optimal SO offers a better encoding complexity and latency than the one with optimal SO, as we will se later in Section\,\ref{subA}. In the second scheme, we extract the \textit{random seed} needed in the first scheme from the secret shares of a few users, while ensuring that the weak secrecy condition, specified in \eqref{sec-con}, is still guaranteed. In the resulting protocol, no external randomness is required, and hence, the total size of data to be stored on storage nodes is equal to the total size of secrets. This shows that the minimum achievable storage overhead is one and the second scheme achieves it. Furthermore, the developed method is applied to communication-optimal DSSPs constructed in Section\,\ref{sec:four}, thereby providing DSSPs that have both the optimal communication complexity and the optimal storage overhead under the weak secrecy constraint.

 \subsection{DSSP with Nearly Optimal Storage Overhead} \label{subA}
 
Let $m$ and $n$ denote the number of users and storage nodes, respectively. Consider a system with access structure $\sA$ consisting of $m$ subsets $A_j$, for $j \in [m]$, of $[n]$ forming a Sperner family. Such a condition is necessary, by \Lref{spe-con}, in order to ensure existence of a valid DSSP. This condition is also sufficient for the proposed protocol in this Section. Note that in the construction of DSSPs we consider the access structure to be a part of the protocol design, as defined in Definition\,\ref{ssd}. However, it is still interesting to see necessary and/or sufficient conditions on the access structure in the proposed DSSPs. Also, suppose that $q > \max_{j \in [m]} |A_j|$ and let $\gamma_1,\gamma_2,\dots,\gamma_{q-1}$
denote the non-zero elements of $\Fq$.

First, consider a S-DSSP protocol, as described in Section\,\ref{sec:three}, which applies Shamir's secret sharing method to each secret independently. 
In particular, for $j \in [m]$ and $l \in [k-1]$, where $k = |A_j|$, $p_{j,l}$ is chosen independently and uniformly at random from $\Fq$. Then the polynomial $P_j(X)$ is constructed as  
\begin{equation}\label{polynomial}
P_j{(x)}=s_j+ \sum_{l=1}^{k-1} p_{j,l}x^l.
\end{equation}
Then the evaluations $P_j(\gamma_1),...,P_j(\gamma_k)$ are stored at nodes indexed by elements in $A_j$, the access set of user $j$. In this protocol the storage overhead is $\sum_{j=1}^m |A_j|/m$. Note that this can be much larger than one which is the lower bound on the storage overhead, as stated in Section\,\ref{sec:two}.


In order to reduce the storage overhead, the idea is to ensure that for all $j \in [m]$, the evaluation of $P_j$'s over the evaluation points $\gamma_i$'s have significant overlaps with each other. This idea is elaborated through the rest of this section.

Suppose that $n$ symbols, denoted by $y_1,y_2,\hdots,y_n$, are chosen independently and uniformly at random from 
$\Fq$. Roughly speaking, these symbols serve as \textit{random seed} in our proposed protocol. Initially, $y_i$ is stored in the storage node $i$, for all $i \in [n]$. Then for each user $j$, only one data symbol is generated through the encoding of its secret $s_j$, as discussed next, and stored in one of the nodes in $A_j$. This implies that the total size of data stored at the storage nodes, in terms of the number of $\F_q$-symbols, is $n+m$, and the storage overhead is $1+\frac{n}{m}$.

Consider user $j$, for some $j \in [m]$. Without loss of generality suppose that $A_j = \{1,2,\hdots,k\}$. To encode the secret $s_j$, we first construct $P_j$ by considering the following system of linear equations: 
 \begin{equation}
 \label{Pj-eq}
  \left\{
   \begin{array}{llll } 
   P_{j}(\gamma_1)=y_1,\\
   P_{j}(\gamma_2)=y_2,\\
   \quad \vdots& \\
   P_{j}(\gamma_{k-1})=y_{k-1},
     \end{array} \right.
  \end{equation} 
which can be rewritten as
 \begin{equation}
 \label{eq_p}
   \left \{
   \begin{array}{l}
   s_j+p_{j,1}\gamma_1+p_{j,2}\gamma_1^{2}+\dots+p_{j,k-1}\gamma_1^{k-1}=y_1,\\
    s_j+p_{j,1}\gamma_2+p_{j,2}\gamma_2^{2}+\dots+p_{j,k-1}\gamma_2^{k-1}=y_2,\\
    \\
    \vdots\\
         s_j+p_{j,1}\gamma_{k-1}+p_{j,2}\gamma_{k-1}^{2}+\dots+p_{j,k-1}\gamma_{k-1}^{k-1}=y_{{k-1}},
    \\
 
   \end{array}\right.
   \end{equation}
and alternatively, in the matrix form, as

   \begin{equation}\label{system}
   \bold{D} \ \bold{p}_j=\tilde{\bold{y}}_j,
   \end{equation}
   where 
$$
\bold{p}_j = (p_{j,1},p_{j,2},\dots,p_{j,k-1})^T, 
$$
\be{yjtilde}
\tilde{\bold{y}}_j = (y_1-s_j,y_2-s_j,\dots,y_{{k-1}}-s_j)^T,
\ee
and  
\be{bd_eq}
\bold{D}\ \deff \
\begin{bmatrix}
\gamma_1 & \gamma_1^2 & \dots & \gamma_1^{k-1}\\
\gamma_2& \gamma_2^2 & \dots & \gamma_2^{k-1}\\
\vdots&&\ddots&\vdots\\
\gamma_{k-1} & \gamma_{k-1}^2 & \dots & \gamma_{k-1}^{k-1}\\
\end{bmatrix}.
\ee
Since $\bold{D}$ is a Vandermonde matrix, it is invertible and \eq{system} has a unique solution, i.e.,
\be{pj_eq}
\bold{p}_j=\bold{D}^{-1}\tilde{\bold{y}}_j.
\ee
By selecting such a $\bold{p}_j$ as the vector of coefficients of $P_j$, except its constant coefficient which is equal to $s_j$, we somewhat \textit{enforce} the first $k-1$ evaluations of $P_j$ to be equal to the random seed generated a priori. Then $P_{j}(\gamma_{k})=y_{j+n}$ is the new data symbol generated in the encoding process of $s_j$. We refer to $y_{j+n}$ as the data symbol associated with user $j$. This encoding process can be also described using the following equation:

\begin{equation}
\label{eq_newdata}
y_{j+n}=s_j+\boldsymbol{a}^{t}\ \tilde{\bold{y}}_j,
\end{equation}
where 
\begin{equation}
\label{adeff}
\boldsymbol{a}^{t}=
\left [
\begin{array}{llll}
\gamma_k&
\gamma_k^2&
\dots&
\gamma_k^{k-1}
\end{array}\right  ]
\bold{D}^{-1},
\left.
\right.
\end{equation}
and $\tilde{\bold{y}}_j$ and $\bold{D}$ are defined in \eq{yjtilde} and \eq{bd_eq}, respectively.

The new data symbol $y_{j+n}$, determined in \eqref{eq_newdata}, is stored in the storage node $k$, the node with the largest index in $A_j$. Note that we could initially pick any $k-1$ nodes in $A_j$ and use the random seed stored therein to construct $P_j$, as specified by \eqref{Pj-eq}. Then the generated data symbol $y_{j+n}$ would be stored in the remaining node of $A_j$. In the next section, we discuss methods to do this in a more structured way in order to obtain a \textit{balanced} storage profile. 

The proposed protocol together with the storage profile are demonstrated in an example, discussed next.

\begin{exmp}\label{nearlyoptimalDSSPexmple}

Let $n=5, m=10$.  
Suppose that the access sets are all $2$-subsets of $[5]$ given as follows:
\begin{align*}
&A_1\hspace{-1mm}=\hspace{-1mm}\{1,2\},\hspace{-.3mm}A_2\hspace{-1mm}=\hspace{-1mm}\{2,4\},\hspace{-.3mm}A_3\hspace{-1mm}=\hspace{-1mm}\{3,5\},\hspace{-.3mm}A_4\hspace{-1mm}=\hspace{-1mm}\{1,3\},\hspace{-.3mm}A_{5}\hspace{-1mm}=\hspace{-1mm}\{2,5\},\\
&A_6\hspace{-1mm}=\hspace{-1mm}\{4,5\},\hspace{-.3mm}A_7\hspace{-1mm}=\hspace{-1mm}\{2,3\},\hspace{-.3mm}A_8\hspace{-1mm}=\hspace{-1mm}\{1,5\},\hspace{-.3mm}A_9\hspace{-1mm}=\hspace{-1mm}\{3,4\},\hspace{-.3mm}A_{10}\hspace{-1mm}=\hspace{-1mm}\{1,4\}.
\end{align*}
Also, let $q=3$, and non-zero and distinct evaluation points $\gamma_1=1,\gamma_2=2$, as elements of $\F_3$, are considered.
Then the encoded data symbols, generated by \eqref{eq_newdata}, together with the storage profile are shown in Table \ref{nearlyoptimalexample}.

\begin{table}[h!]
\centering
 \begin{tabular}{||c c c c c||} 
 \hline
 Node 1 & Node 2 & Node 3 & Node 4 &Node 5 \\ [0.5ex] 
 \hline\hline
 $y_1$ & $y_2$ & $y_3$ & $y_4$ & $y_5$ \\ 
 \hline
 $ $ & $2y_1-s_1$ & $ 2y_1-s_4$ & $2y_2-s_2 $ & $ 2y_3-s_3$ \\
 \hline
  &  & $2y_2-s_7$ &$2y_3-s_9$& $2y_2-s_5$ \\
 \hline
  &  &  & $2y_1-s_{10}$&$2y_4-s_6$ \\
 \hline
  &  &  & & $2y_1-s_8$\\ [1ex] 
 \hline
\end{tabular}
\caption{ Storage Profile in Example \ref{nearlyoptimalDSSPexmple}}
\label{nearlyoptimalexample}
\end{table}
\end{exmp}

To prove that the weak secrecy condition, specified in \eq{sec-con}, is satisfied the following lemma is useful.

\begin{lemma}\label{fulent}
Data symbols $y_1,y_2,\dots,y_{m+n}$ generated according to the proposed protocol are uniformly distributed and mutually independent. In other words, the vector of all data symbols is full entropy.

\end{lemma}

\begin{proof}
Recall that the first $n$ data symbols $y_1,y_2,\dots,y_n$ are initially selected independently and uniformly at random. Hence, $(y_1,y_2,\hdots,y_n)$ is a full entropy vector. Then
\begin{align}
\label{fulent-eq1}
H(y_{n+1},\hdots,y_{n+m} |y_1,\hdots,y_n) &=H(s_1,\hdots,s_m |y_1,\hdots,y_n) \\ 
\label{fulent-eq2}
&=H(s_1,\hdots,s_m) \\
\label{fulent-eq3}
&= m \log q,
\end{align}
where \eq{fulent-eq1} holds since \eqref{eq_newdata} implies that given $(y_1,\hdots,y_n)$ there is a one-to-one mapping between $(y_{n+1},\hdots,y_{n+m})$ and $(s_1,\hdots,s_m)$, \eqref{fulent-eq2} holds since the random seed $(y_1,\hdots,y_n)$ is independent of $(s_1,\hdots,s_m)$, and \eqref{fulent-eq3} holds since it is assumed that the vector of all secrets is full entropy. Using this together with the chain rule we have
\begin{align*}
H(y_1,\hdots y_{n+m})=&H(y_1,\hdots,y_n)\\&+ H(y_{n+1},\hdots,y_{n+m} |y_1,\hdots,y_n)\\
=& n\log q + m \log q\\
=& (n+m)\log q,
\end{align*}
which completes the proof.

\end{proof}

The following theorem summarizes the results of this subsection.

\begin{theorem}\label{nearlyoptimalsecurity}
The proposed protocol in this Section is a weakly secure DSSP satisfying all conditions in Definition \ref{ssd}. 
\end{theorem}

\begin{proof} Note that user $j$, with $|A_j| = k$, has access to all $k$ evaluations of its associated polynomial $P_j$ and, consequently, can recover $s_j$ by invoking Shamir's secret decoder. Hence, the correctness condition is satisfied in this protocol. What remains to show is that the weak secrecy condition, specified in \eqref{sec-con}, is also satisfied.

Note that the access sets are assumed to form a Sperner family. Hence, for $j\neq l$, with $|A_j|=k$, there exists at least one $\gamma_i$, $i \in [k]$, such that $P_j(\gamma_i)$ is not accessed by user $l$. Let this data symbol $P_j(\gamma_i)$ be denoted by $y_j^{(-l)}$. Then we have
\begin{equation}
\begin{split}
\label{thm1-eq1}
H(s_j |\bold{y}_l) \stackrel{\text{(a)}}{\geq}\ & H(s_j| \ \bold{y} \ \setminus  \ y_j^{(-l)}) \stackrel{\text{(b)}}= H(y_j^{(-l)} | \ \bold{y} \ \setminus  \ y_j^{(-l)}) \\
&\stackrel{\text{(c)}}= H(y_j^{(-l)})\stackrel{\text{(d)}} = \log q,
\end{split}
\end{equation}
 where (a) holds since conditioning does not increase the entropy, (b) holds because given any $k-1$ evaluations of $P_j$, out of $k$ available ones, there is a one-to-one mapping between the remaining evaluation of $P_j$ and $s_j$, (c) holds because data symbols are independent, and (d) holds because data symbols are uniformly distributed. Also, note that
\begin{equation}
\label{thm1-eq2}
H(s_j |\bold{y}_l) \leq H(s_j) \leq \log q. 
\end{equation}
Combining \eqref{thm1-eq1} with \eqref{thm1-eq2} implies that
$$
H(s_j |\bold{y}_l)=H(s_j)=\log q,
$$
which completes the proof. 

\end{proof}

As mentioned before, the storage overhead of the proposed protocol is $1+\frac{n}{m}$. This is very close to the optimal value, which is shown to be equal to one in Section\,\ref{optimaldssp}, provided that $m$ is much larger than $n$. Hence, we refer to the proposed protocol in this subsection as the DSSP with nearly optimal storage overhead.  
\subsection{DSSP with Optimal Storage Overhead} \label{optimaldssp}

The parameters are same as those considered in Section\,\ref{subA}. However, an extra condition on the field size $q$ and a certain choice of access sets for $n$ of the users, as explained later, will be required in the proposed protocol. 

In the proposed protocol in this section, no external randomness is required and in fact, the optimal storage overhead equal to one is attained. The idea is to modify the proposed DSSP with nearly optimal SO, constructed in Section\,\ref{subA}, by initially encoding secrets of $n$ users into $n$ data symbols and then utilizing them as the random seed required in the encoding of the remaining $m-n$ secrets.


To initialize the protocol, $n$ users, indexed by $1,2,\dots,n$, are picked. The access sets for these users are specified as follows: 
\be{cyclicset}
A_j=\{j,j+1,\hdots,j+k-1\},
\ee
for $j=1,2,\dots,n$, where the indices of storage nodes are considered modulo $n$, i.e., ${n+l} = l$. The parameter $k$ is arbitrary as long as $1\leq k < n$. When combining DSSP with optimal storage overhead toegther with communication-optimal DSSP, $k$ needs to be picked accordingly, e.g., $k = i$ or $k=i+1$, where $i$ is given by \Tref{sec:four_main}. The remaining access sets can be arbitrary as long as $\sA$ form a Sperner family.

The evaluation polynomials $P_j$'s, for $j \in [n]$, are constructed in such a way that the following system of linear equations is satisfied:
 \begin{equation}
 \label{eq_SO}
  \left\{
   \begin{array}{llll } 
   P_{1}(\gamma_1)=y_1,&  P_{2}(\gamma_1)=y_2,& \ldots & P_{n}(\gamma_1)=y_n,\\
   P_{1}(\gamma_2)=y_2,&  P_{2}(\gamma_2)=y_3,&\ldots & P_{n}(\gamma_2)=y_1,\\
   \quad \vdots& \quad \vdots& \ddots &  \quad \vdots \\
   P_{1}(\gamma_{k})=y_{k},&  P_{2}(\gamma_{k})=y_{k+1},&\ldots & P_{n}(\gamma_{k})=y_{k-1},
     \end{array} \right.
  \end{equation}  
where $P_j$'s have the form specified by \eqref{polynomial}. Note that the system in \eq{eq_SO} has $nk$ linear equations. Also, the total number of variables is $nk$. This is because there are $(k-1)n$ unknown variables $p_{j,l}$'s, for $j \in [n]$ and $l \in [k-1]$, together with $n$ unknown variables $y_1,y_2,\dots,y_n$. Next, we show that the equations in \eq{eq_SO} are linearly independent, under certain conditions, thereby establishing that the system has a unique solution for $p_{j,l}$'s and $ y_i$'s. We further show how to store the resulting $y_i$'s in the storage nodes according to a certain access structure $\sA$. 

The system of linear equations in \eq{eq_SO} can be rewritten as:
 \begin{equation}\label{matrixform}
 \bold{A} \bold{b}+\bold{s}'=0,
 \end{equation}
 where 
\begin{align*}
 \bold{b}&=(p_{1,1},...,p_{1,k},p_{2,1},...,p_{2,k},...,p_{n,1},...,p_{n,k},y_1,...,y_n)^T,\\
\bold{s}'&=(s_1,...,s_1,s_2,...,s_2,...,s_n,...,s_n)^T,
 \end{align*}
 where each $s_j$ is repeated $k$ times in $\bold{s}'$. Also, $\bold{A}_{(kn)\times( kn)}$ is equal to
\be{matrixA}
\left [
   \setlength{\arraycolsep}{3pt}
   \renewcommand{\arraystretch}{0.7}
\begin{array}{cccccccc|ccccc}
    \gamma_{1}& \dots &\gamma_{1}^{k-1} &  &  &  &   &&  -1 & &&&\\
    \vdots &  \ddots & \vdots&  & &  &&&& \ddots && &  \\
    \gamma_{k}& \dots & \gamma_{k}^{k-1}&&&&& &&  & -1&&\\
     &&&\gamma_{1} &  \dots &\gamma_{1}^{k-1}&&&&-1 &  && \\
    &&&\vdots &  \ddots & \vdots&&  &&& \ddots && \\
    &&&\gamma_{k}&  \dots & \gamma_{k}^{k-1}& & && & & -1 &\\
    &&&&&&\ddots&&&&&\ddots&\\
    &&&&&&&&&&&&\\
    &&&&&&&&&&&&
\end{array}
\right],
\ee
which consists of $n$ copies of the following $k \times(k-1)$ matrix
\[
\bold{B} \ \deff\
\begin{bmatrix}
    \gamma_{1} & \gamma_{1}^2 & \dots &\gamma_{1}^{k-1} \\
    \gamma_{2} & \gamma_{2}^2 & \dots &\gamma_{2}^{k-1} \\
    \vdots & \vdots & \ddots & \vdots \\
    \gamma_{k}& \gamma_{k}^2 & \dots & \gamma_{k}^{k-1}  
\end{bmatrix},
\]
together with $n$ copies of $-I_{k\times k}$, each shifted one column to the right, consecutively. The goal is then to show that $\bold{A} $ is non-singular, under certain conditions, as specified in the next lemma. To simplify such conditions and also for ease of calculation, let $\gamma_i = \gamma^i$, for $i = 1,2,\dots,k$, where $\gamma$ is a primitive element of $\Fq$. 

\begin{lemma}\label{nonsingular}
 If $ (q-1) \notdivides in$ for $i \in [k]$, then the matrix $\bold{A}$, specified in \eq{matrixA}, is non-singular. 
\end{lemma}
 \begin{proof}
 Let $\bold{r}_{i,j}$ denote the row in $\bold{A}$ indexed by $(j-1)k+i$, for $j \in [n]$ and $i \in [k]$. We show that $\bold{r}_{i,j}$'s are linearly independent. Suppose that a linear combination of $\bold{r}_{i,j}$'s is zero, i.e.,
 \begin{equation*}
 \lambda_{1,1}\bold{r}_{1,1}+\dots +\lambda_{k,1}\bold{r}_{k,1}+\dots+\lambda_{1,n}\bold{r}_{1,n}+\dots +\lambda_{k,n}\bold{r}_{k,n}=0.
 \end{equation*}
 Hence,
 \begin{equation}\label{recsys}
 \bold{B}^T \bold{\lambda}_i=\bold{0},
 \end{equation}
 where $ \bold{\lambda}_i=(\lambda_{i,1},\lambda_{i,2},\dots,\lambda_{i,k})^T $ for all $i$ and
 \[
\bold{B}^T =
\begin{bmatrix}
\gamma & \gamma^2 & \dots & \gamma^k\\
\gamma^2 & (\gamma^2)^2& \dots &(\gamma^2)^k\\
&&&\\
\vdots &&&\\
&&&\\
\gamma^{k-1} &(\gamma^{k-1})^2 &\dots&(\gamma^{k-1})^k
\end{bmatrix}.
\]
Furthermore,
  \begin{equation}\label{sumzero}
   \left\{
   \begin{array}{l} 
   \lambda_{1,1}+\lambda_{2,n}+\lambda_{3,n-1}+\dots+\lambda_{k,n-k+1}=0,\\
   \lambda_{1,2}+\lambda_{2,1}+\lambda_{3,n}+\dots+\lambda_{k,n-k+2}=0,\\
   \vdots \\
   \lambda_{1,n-1}+\lambda_{2,n-2}+\lambda_{3,n-3}+\dots+\lambda_{k,n-k}=0,\\
   \lambda_{1,n}+\lambda_{2,n-1}+\lambda_{3,n-2}+\dots+\lambda_{k,n-k+1}=0.
   \end{array} \right.
  \end{equation}   
 Since $\bold{B}^T$ is a Vandermonde matrix, it is full row rank. Consequently, its kernel space is one dimensional. This together with 
 \eq{recsys} result in $\bold{\lambda}_j=\eta_j \bold{v}$, where $\eta_j$ is a scalar coefficient and $\bold{v}$ is a non-zero vector in the kernel of $\bold{B}^T$. Let $\bold{v} = (v_1,v_2,\dots,v_k)^T$. Then, one can write $\lambda_{i,j}=\eta_j v_i$ for all $i$ and $j$. Substituting this in \eq{sumzero} results in:
 \begin{equation*}
 \bold{V} \bold{\eta}=\bold{0},
 \end{equation*}
 where $\bold{\eta}=(\eta_1,\eta_2,\dots,\eta_n)^T$ and
 \[
\bold{V} \ \deff\
\begin{bmatrix}
    v_1 & 0 & 0& \dots &v_k & v_{k-1}&\dots &v_2\\
    v_2&v_1&0&\dots&0&v_k&\dots&v_3 \\
    \\
    \vdots&&&&\ddots&&&\vdots\\
    \\
    0&\dots&0&v_k&\dots&v_1&\dots&0
\end{bmatrix}_{n\times n} .
\]
If $\bold{V}$ is non-singular, all $\eta_j$'s are zero. This implies that all $\lambda_{i,j}$'s are also zero. Note that all $v_i$'s can not be zero because the kernel space of $\bold{B}^T$ is one dimensional. Therefore, $\bold{A}$ is non-singular if and only if $\bold{V}$ is non-singular. Note that $\bold{V} $ is a circulant matrix and is non-singular if and only if $\gcd(x^n-1,V{(x)})=1$ \cite{geller2004circulant}, where $V(X)$ is the associated polynomial of the circulant matrix $\bold{V}$:
 \begin{equation*}
 V(x)=v_1+v_2x+v_3x^2+\dots+v_kx^{k-1}. 
 \end{equation*}
Note that $\bold{B}^T\bold{v}=0$. Therefore, $\gamma^iV(\gamma^i)=0$, for $i \in [k-1]$. Equivalently, all $\gamma^i$'s are roots of $xV{(x)}$. Since the degree of $V{(x)}$ is at most $(k-1)$, it has at most $(k-1)$ roots. Therefore, we can write
\begin{equation*}
V{(x)}=c_0 (x-\gamma)(x-\gamma^2)\dots(x-\gamma^{k-1}),
\end{equation*}
for some constant $c_0$.
Since $\gamma$ is a primitive element of $\mathbb{F}_q$ and $ (q-1) \notdivides in$ for all $i \in [k]$, then $\gamma^{in} \neq 1$ for all $1\leq i \leq k$. In other words, none of $V(X)$'s roots is an $n$-th root of unity. 
Hence, $x^n-1$ and $V{(x)}$ have no roots in common implying that $\gcd\bigl(x^n-1,V{(x)}\bigr)=1$. Consequently, $\bold{V}$ is non-singular. This implies that $\bold{A}$ is also non-singular, which completes the proof.
 \end{proof} 
 
 \begin{corollary}
 \label{cor_onetoone}
 If the condition in \Lref{nonsingular} is satisfied, then \eq{eq_SO} defines a one-to-one mapping between $(s_1,s_2,\dots,s_n)$ and $(y_1,y_2,\dots,y_n)$.
\end{corollary}
\begin{proof}
Since the condition in \Lref{nonsingular} is satisfied, then $\bA$ is non-singular. This implies that for any vector of secrets $s_j$'s, there is a unique solution for $y_j$'s. Furthermore, for given $y_j$'s, \eq{eq_SO} defines $n$ interpolation equations of polynomials $P_j(x)$ of degree at most $k-1$, for which there exists a unique solution. Hence, there exists a one-to-one mapping between $(s_1,s_2,\dots,s_n)$ and $(y_1,y_2,\dots,y_n)$.
\end{proof}

Suppose that the condition in \Lref{nonsingular} is satisfied, e.g., $q > kn+1$. Then \eq{matrixform} can be written as $\bold{b}=-\bold{A}^{-1}\bold{s'}$. Note that $\bold{b}$ contains all data symbols $y_1,y_2,\dots,y_n$ as its last $n$ entries. Also, $\bold{s}'=\bold{K} \ (s_1,\hdots,s_n)^T $, where 
\[
\bold{K} =
\begin{bmatrix}
\left.
\begin{array}{l}
1\\
\vdots 
\\
1
\end{array}\right\}k

\\

&\left.\begin{array}{l}
1\\
\vdots 
\\
1
\end{array}\right\}k&&&&&&\\

\\
&&&&\ddots&&&\\
\\
&&&&&&& \begin{array}{l}

1\\
\vdots 
\\
1
\end{array}\\

\end{bmatrix}_{kn\times n}.
\] 

\noindent
Let $\bold{A}'$ be a submatrix of $-\boldsymbol{A}^{-1}$ consisting of the last $n$ rows of $-\boldsymbol{A} ^{-1}$. Now we can write the encoding equation for users $1,\hdots,n$ as follows:
\begin{equation}\label{encoding}
(y_1,\hdots,y_n)^T=\bold{E} \ (s_1,\hdots,s_n)^T,
\end{equation}
where $\bold{E}\ \deff \ \bold{A}'\bold{K}$ and is referred to as the \emph{seed encoding matrix}. Note that by \Cref{cor_onetoone}, there is a one-to-one mapping between $(s_1,s_2,\hdots,s_n)$ and $(y_1,y_2,\hdots,y_n)$. Hence, $\bold{E}$ is non-singular. 

In order to attain the storage overhead equal to one, each data symbol $y_j$ must be stored only once while ensuring that the correctness condition for users $1,2,\dots,n$ is satisfied. For $j \in [n]$, $y_j$ is stored at node $j$. Hence, by the choice of $A_j$'s specified in \eq{cyclicset} together with \eq{eq_SO}, user $j$ has access to all $k$ realizations of $P_j$ and by invoking Shamir's secret decoder, it can reconstruct $s_j$.

The encoding of secrets $s_j$, for $j=n+1,\dots,m$, is exactly same as in the DSSP with nearly optimal SO, discussed in Section\,\ref{subA}. In other words, the same protocol with $m-n$ users and assuming $y_1,y_2,\dots,y_n$, constructed from secrets $s_1,s_2,\dots,s_n$ as discussed above, as the initial random seed is invoked. In \Tref{optimalso} it is proved that the proposed protocol is a DSSP satisfying the weak secrecy condition, specified in \eqref{sec-con}, and hence, it is referred to as the weakly secure DSSP with optimal SO. Before that, a construction of such DSSP is demonstrated in the following example. 
\begin{exmp}
\label{optimalDSSPexmple}

 Let $m=10$ and $n=5$, as in Example \ref{nearlyoptimalDSSPexmple}. Let also $k=2$ and $q=5$, which satisfy the condition of \Lref{nonsingular}. Then $5$ subsets with access sets specified in \eqref{cyclicset} for users indexed by $1,\hdots,5$ are picked. In particular, the following access sets are considered for all the $10$ users:
\begin{align*}
&A_1 \hspace{-1mm}=\{1 ,2\},A_2\hspace{-1mm}=\hspace{-1mm}\{2 ,3\},A_3\hspace{-1mm}=\hspace{-1mm}\{3 ,4\},A_4\hspace{-1mm}=\{4 ,5\},A_5\hspace{-1mm}=\{5 ,1\},\\
&A_6\hspace{-1mm}=\{1,3\},A_7\hspace{-1mm}=\{1,4\},A_8\hspace{-1mm}=\{1,5\},A_9\hspace{-1mm}=\{2,4\},A_{10}\hspace{-1mm}=\{3,5\},
\end{align*}
where $y_6,\hdots,y_{10}$ are stored on storage nodes indexed by the largest elements of $A_6,\hdots,A_{10}$, respectively. Let $\gamma = 2$ be picked as a primitive element of $\mathbb{F}_5$ and then $\gamma_1=2, \gamma_2=4$ are the evaluation points. In the first step of the protocol, which involves encoding $s_1,\hdots,s_5$ as the random seed, the \textit{seed encoding} matrix $\bold{E}$ in \eqref{encoding} is as follows:
 
 \[
\bold{E} \ = \
\begin{bmatrix}
 1&3&4&2&1\\
 1&1&3&4&2\\
 2&1&1&3&4\\
 4&2&1&1&3\\
 3&4&2&1&1
\end{bmatrix} .
\] 
Then, $y_1,\hdots,y_5$ are computed from $s_1,\dots,s_5$ as
\begin{align*}
y_1=s_1+3s_2+4s_3+2s_4+s_5,\\
y_2=s_1+s_2+3s_3+4s_4+2s_5,\\
y_3=2s_1+ s_2+ s_3+3s_4+4s_5,\\
y_4=4s_1+2s_2+ s_3+ s_4+3s_5,\\
y_5=3s_1+4s_2+2s_3+ s_4+s_5,
\end{align*}
and are used to encode remaining secrets according to the DSSP with nearly optimal SO. The encoded data symbols together with the storage profile are shown in the following table:

\begin{table}[h!]
\centering
 \begin{tabular}{||c c c c c||} 
 \hline
 Node 1 & Node 2 & Node 3  & Node 4 &Node 5 \\ [0.5ex] 
 \hline\hline
 $y_1$ & $y_2$ & $y_3$ & $y_4$ & $y_5$ \\ 
 \hline
 $ $ & $  $ & $ 2y_1-s_6 $ & $2y_1-s_7  $ & $2y_1-s_8  $ \\
 \hline
  &  & $ $ &$2y_2-s_9 $& $2y_3-s_{10} $ \\
 \hline
\end{tabular}
\caption{ Storage Profile in Example \ref{optimalDSSPexmple}}
\label{optimalexample}
\end{table}

\end{exmp}


\begin{theorem}\label{optimalso}
The proposed protocol in this section is a weakly secure DSSP satisfying all conditions in Definition \ref{ssd} and has the storage overhead, defined in \eqref{so}, equal to $1$.
\end{theorem}
\begin{proof}
In this protocol, each user $j$ has access to all $|A_j|$ evaluations of its associated polynomial $P_j$, which holds for both the initial $n$ users and the remaining $m-n$ users. Hence, the correctness condition is satisfied by invoking Shamir's secret decoder. Also, note that the number of data symbols generated in this protocol is $n$, equal to the number of users, and each data symbol is stored exactly once. Hence, the storage overhead is one, which is the optimal value. What remains to show is that the weak secrecy condition is also satisfied.

Since the vector of all secrets is assumed to be full entropy, then $(s_1,s_2,\hdots,s_n)$ is also full entropy. This together with \Cref{cor_onetoone} implies that $(y_1,y_2,\hdots,y_n)$ is full entropy and independent of $(s_{n+1},s_{n+2},\hdots,s_{m})$. Also, by \Lref{fulent}, the vector of data symbols $(y_{n+1},y_2,\hdots,y_m)$ is full entropy. Consequently, the vector of all data symbols in this protocol is full entropy. Then the rest of the proof is similar to the proof of \Tref{nearlyoptimalsecurity}. 

\end{proof}

\noindent
{\bf Remark\,1.} Note that we always have SO $\geq 1$ in a DSSP to satisfy the correctness condition, as stated in Section\,\ref{sec:two}. In fact, \Tref{optimalso} implies that this lower bound is achievable under a weak secrecy condition. Hence, we refer to the proposed protocol in this section as the DSSP with optimal storage overhead. 

For certain parameters $m$ and $n$, the proposed protocol in this section is also a communication-optimal DSSP. This is summarized in the next Theorem. 

\begin{theorem}
\label{optimal}
Let $m= {n\choose k}$ for some $k\leq n/2$ and the access structure $\sA$ be picked as the set of all $k$-subsets of $[n]$. Then the DSSP with optimal storage overhead is also communication-optimal. In other words, it simultaneously attains the optimal value for both the communication complexity and the storage overhead under the weak secrecy condition, specified in \eqref{sec-con}. 
\end{theorem}

\begin{proof}
 Note that the DSSP with optimal SO is a T-DSSP, as defined in Definition \ref{tdss}. This is because a user $j$ downloads exactly one data symbol from each node in $A_j$ to reconstruct $s_j$. Also, $\sA$ is the collection of all $k$-subsets of $n$ and, in particular, the ones specified in \eq{cyclicset}. Then the conditions in \Tref{sec:four_main} is satisfied, which implies that the proposed protocol is also a communication-optimal DSSP.
\end{proof}

An interesting case of \Tref{optimal} is when we want to serve the maximum possible number of users $m={n\choose \floor{\frac{n}{2}}}$ for a given $n$, as stated in Section\,\ref{sec:three}. In this case, we have a communication-optimal DSSP with optimal storage overhead that also serves the maximum possible number of users while providing weak secrecy, as specified in \eqref{sec-con}.


Next, we discuss the complexity of the construction and encoding as well as the latency of the encoding process. The complexity of constructing the DSSP with optimal SO is dominated by computing the inverse matrix $\bold{A}^{-1}$. The complexity of a straightforward Gaussian elimination method for computing $\bold{A}^{-1}$ is $O(k^3n^3)$. Note that this needs to be done only once, and then $\bold{A}$ can be fixed for encoding purposes. The computation complexity of encoding the first $n$ secrets is $O(n^2)$, due to the multiplication of the $n \times n$ \textit{seed encoding matrix} by the vector of secrets of the first $n$ users, as specified in \eqref{encoding}. In the second step of the encoding process, the computation of a vector inner product, as specified in \eq{eq_newdata}, is needed for each of the $m-n$ remaining users resulting in a complexity $O((m-n)\tilde{k})$, where $\tilde{k}$ is the average size of access sets of the remaining $m-n$ users. Note that one can assume the vector $\boldsymbol{a}^{t}$ is computed a priori, as part of the construction, and hence, it does not have to be computed during the encoding process. Hence, a straightforward implementation of the encoding process results in the encoding complexity $O(n^2+\tilde{k}m)$. Similarly, the computation complexity of the encoder in the DSSP with nearly optimal SO is $O(\tilde{k}m)$, where $\tilde{k}$ is the average size of access sets of all the $m$ users. Moreover, encoding of all secrets in this protocol can be done in parallel, which results in a \textit{fast} encoder implementation with latency $O(\max_{j} |A_j|) = O(n)$. However, in a DSSP with optimal SO, encoding of the secrets of the first $n$ users should be done first, followed by encoding secrets of the $m-n$ remaining users, which results in $O(n^2)$ latency. Table \ref{comparison} summarizes the complexity comparison of the two protocols proposed in this section. Although the DSSP with nearly optimal SO has a slightly higher SO than the optimal one, it has a much lower encoding latency and also slightly lower encoding complexity. Also, the DSSP with optimal SO requires a condition on the field size $q$ specified in \Lref{nonsingular}, in addition to the natural condition $q > \max_{j \in [m]} |A_j|$. Regarding the decoding complexity, note that in both of the protocols under discussion, each user $j$ utilizes Shamir's decoder in which a polynomial is interpolated to reconstruct $s_j$. Note also that the coefficients of the interpolated polynomial can be computed once and be used repeatedly for the decoding. Hence, the decoding complexity of both protocols is $O(\tilde{k}m)$. Also, users can decode their secrets in parallel resulting in $O(\tilde{k})$ decoding latency.   

\vspace{3mm}
\begin{table}
\begin{center}

\begin{tabular}{ |p{1.7cm}||p{3.3cm}|p{2.5cm}|  }
 \hline
 \vspace{.1mm}Protocol & \vspace{.1mm} DSSP  with optimal SO & DSSP with nearly optimal SO\\
 \hline
storage overhead \vspace{1mm}&  \vspace{0.4mm} \hspace{17mm}$1$ & \vspace{0.4mm}$\hspace{8mm}1+\frac{n}{m}$ \\
 \hline
 \vspace{0.05mm}encoding complexity\vspace{1mm} &\vspace{.1mm}\hspace{8mm}$O(n^2+\tilde{k}m)$ \vspace{.1mm}&\vspace{.1mm}$\hspace{7mm} O(\tilde{k}m)$\\
 \hline
 \vspace{0.4mm}
encoding latency  \vspace{1mm}&\vspace{.1mm}\hspace{13mm}$O(n^2)$ \vspace{0mm}& \vspace{0mm}$\hspace{8mm}O(n)$\\
 \hline
 \vspace{0.05mm}decoding complexity\vspace{1mm} &\vspace{.1mm}\hspace{13mm}$O(\tilde{k}m)$ \vspace{.1mm}&\vspace{.1mm}$\hspace{7mm} O(\tilde{k}m)$\\
 \hline
  \vspace{0.4mm}
decoding latency  \vspace{1mm}&\vspace{.1mm}\hspace{13mm}$O(\tilde{k})$ \vspace{0mm}& \vspace{0mm}$\hspace{8mm}O(\tilde{k})$\\
 \hline
\end{tabular}

\end{center}
\caption{ Comparison of DSSPs with optimal and nearly optimal SO.}
\label{comparison}
\end{table}

\section{DSSP with Balanced Communication Complexity and Storage Profile}\label{subC}

In this section, we discuss methods to make the storage profile and the communication complexity \textit{balanced} when considering individual storage loads and communication loads across the storage nodes. 

In practice, the time required to establish access to storage nodes is one of the major factors affecting the performance in large scale distributed storage systems \cite{dean2009challenges,melnik2010dremel}. An unbalanced storage load can lead to non-uniform delays across different nodes and, ultimately, to node failures. Hence, it is highly desirable for a distributed storage system to be able to balance offered data access load across the storage nodes \cite{aktas2019load}. More specifically, assuming that the storage nodes have the same capabilities, it is not desired to have a node with a significantly larger amount of stored data comparing to another node. This becomes relevant when the actual limitations on the storage capacity of nodes is taken into account. Similarly, if a significantly larger amount of data needs to be downloaded from one node comparing to another node, that could slow down the process of reconstructing the secrets. 

Suppose that the \textit{storage profile vector} $\boldsymbol{c}^s=(c^s_1,c^s_2,\hdots,c^s_n)$, corresponding to a certain storage profile, represents the amount of data stored in the storage nodes, where $c^s_i$ is the amount of data, in terms of the number of $\Fq$-symbols, stored in node $i$. Similarly, \textit{communication complexity vector} $\boldsymbol{c}^c=(c^c_1,c^c_2,\hdots,c^c_n)$, corresponding to the collection of decoding processes $\cD$ defined in Definition \ref{ssd}, represents the total amount of data downloaded from the storage nodes to reconstruct all secrets, where $c^c_i$ is the total amount of downloaded data from the node $i$. 
 \begin{definition}\label{balanced-def}
 
 We say that a vector is \textit{balanced} if all entries are equal. Furthermore, we define the \textit{bias} of a vector to be the maximum difference between the entries of the vector, i.e., the difference between the maximum and the minimum entry. In the context of access sets, we say that a collection $\cF$ of subsets of $[n]$ is a \textit{balanced collection} if each $i \in [n]$ belongs to the same number of subsets in $\cF$. 
 
 \end{definition}

 For instance, the strategy used in Example \ref{optimalDSSPexmple} results in the storage profile vector $(1,1,2,3,3)$, as shown in Table \ref{optimalexample}, with a bias of $2$. Ideally, one would want a balanced storage profile vector $(2,2,2,2,2)$, which we later show is actually possible. In general, the goal here is to specify the access structure $\sA$ and the storing matrix $\boldsymbol{Z}$ in order to reduce the bias of storage profile vector and communication complexity vector of DSSPs proposed in Section\,\ref{sec:five}. This, roughly speaking, results in protocols that are less \textit{biased} and more \textit{balanced}. 

The following lemma is a key to enable constructing such DSSPs with balanced properties. 
 \begin{lemma}
 \label{weakpartition}
 The set of all $k$-subsets of $[n]$ can be partitioned into \textit{balanced collections} of size at most $n$. 
 \end{lemma}
 
 \begin{proof}
 Let $\script{K}$ denote the set of all $k$-subsets of $[n]$. For any $A=\{i_1,\hdots,i_k\} \in \script{K}$, the operator $\phi: \script{K} \rightarrow \script{K}$ is defined as follows: 
\be{phiop}
\phi (A)\ \deff \ \{i_1+1,\hdots,i_k+1\},
\ee
with $n+1=1$. Note that 
\be{identity}
\phi^n(A)=A.
\ee
 A relation $\sim$ over $\script{K}$ is defined as follows. We say $A  \sim  B$ if 
 $
 \phi^l(A)=B,
 $ for some integer $l$. It can be observed that $\sim$ is an equivalence relation. Consequently, it partitions $[n]$ into equivalence classes. We show that these equivalence classes are balanced. Note that \eqref{identity} implies that each equivalence class contains at most $n$ distinct elements of $\script{K}$. In fact, it can be shown that an equivalence class $\cF$ has $d$ elements, for some $d$ that divides $n$. If $d = n$, then each $i \in [n]$ belongs to exactly $k$ subsets in $\cF$. Otherwise, for any $A \in \cF$ and $i\in A$, we have
  $$
  \left\{i,i+d,\dots,i+d(\frac{n}{d}-1)\right\} \subseteq A.
  $$
  Hence, $i$ and $i+ld$, for any $i \in [d]$ and $l \in [\frac{n}{d}-1]$, belong to the same number of subsets in $\cF$. Then one can consider all elements of $[n]$ modulo $d$, which reduces $\cF$ to an equivalence class of $\frac{kd}{n}$-subsets of $[d]$ with size $d$. Hence, each $i \in [d]$ belongs to the same number of subsets in $\cF$, which completes the proof.
 \end{proof}

An example of the partitioning discussed in the proof of \Lref{weakpartition} is as follows.

\begin{exmp}\label{partitionexample}
Let $n=6$ and $k=3$. Then, all $3$-subsets of $[6]$ are partitioned into $3$ classes of size $6$ and one class of size $2$, which are all \textit{balanced collections}.
\begin{align*}
&[\hspace{-.5mm}\{1,2,3\}\hspace{-.5mm}]\hspace{-1mm}=\hspace{-1mm}\{ \hspace{-.5mm}\{1,2,3\},\hspace{-1mm}\{2,3,4\},\hspace{-1mm}\{3,4,5\},\hspace{-1mm}\{4,5,6\},\hspace{-1mm}\{5,6,1\},\hspace{-1mm}\{6,1,2\}\hspace{-.5mm}\}\\
&[\hspace{-.5mm}\{1,3,4\}\hspace{-.5mm}]\hspace{-1mm}=\hspace{-1mm}\{ \hspace{-.5mm}\{1,3,4\},\hspace{-1mm}\{2,4,5\},\hspace{-1mm}\{3,5,6\},\hspace{-1mm}\{4,6,1\},\hspace{-1mm}\{5,1,2\},\hspace{-1mm}\{6,2,3\}\hspace{-.5mm}\}\\
&[\hspace{-.5mm}\{1,3,6\}\hspace{-.5mm}]\hspace{-1mm}=\hspace{-1mm}\{ \hspace{-.5mm}\{1,3,6\},\hspace{-1mm}\{2,4,1\},\hspace{-1mm}\{3,5,2\},\hspace{-1mm}\{4,6,3\},\hspace{-1mm}\{5,1,4\},\hspace{-1mm}\{6,2,5\}\hspace{-.5mm}\}\\
&[\hspace{-.5mm}\{1,3,5\}\hspace{-.5mm}]\hspace{-1mm}=\hspace{-1mm}\{ \hspace{-.5mm}\{1,3,5\},\hspace{-1mm}\{2,4,6\}\hspace{-.5mm}\}\\
\end{align*}
\end{exmp}

The partitioning suggested in \Lref{weakpartition} enables constructing DSSPs with optimal storage overhead while having an \textit{almost} balanced communication complexity and storage profile for any $n$ and $m$. However, to this end, we pick the access sets according to a certain process. Same as in Section\,\ref{sec:five}, let $k$ denote the smallest integer with $m \leq { n\choose k} $. Then the process for assigning the access set $A_j$, as a $k$-subsets of $[n]$, for each user $j$ and the storage node in $A_j$ that stores the data symbol $y_j$, generated for each user $j$, is described in Algorithm 1 below.

\begin{algorithm}[H]
\SetAlgoLined
\KwResult{DSSP has balanced $\boldsymbol{c}^c$ and $\boldsymbol{c}^s$}
 \;
  $\script{A}=\emptyset,i=1, j=1, A=\{1,2,\hdots,k\}$, $\boldsymbol{Z}=[0]_{n \times m}$

  \While{$j \leq m$}{
   
   \While{$i\notin A $}{$A=\phi(A)$}

 \While{$A \notin \script{A} $}{
 $A_j=A$, \;
 $\sA=\sA \cup \{A\},$\;
 $z_{i,j}=1$.\;
 
 $A=\phi(A)$,\; $j=j+1$.\;
 
   \eIf{$i < n$}{
   $i=i+1$\;
   
   }{
   $i=1$\;
  }
 }
 $A=$ Any $k$-subset which is not in $\script{A}$ .
 }
\caption{Constructing $\script{A}$ and $\boldsymbol{Z}$ for the DSSP with optimal SO proposed in Section\,\ref{optimaldssp} }
\label{balancedDSSP}
\end{algorithm}

In particular, in Algorithm 1, access sets are picked one-by-one for users $j=1,2,\dots,m$ from equivalence classes by repeatedly applying the operation $\phi$, defined in \eq{phiop}. Through this process, the index of the storing node for each user $j$, the node that stores the one data symbol $y_j$, is also increased one by one. Once all subsets in an equivalence class are picked, the next access set is picked from another equivalence class in such a way that the index of next storing node is also increased by $1$.

Algorithm \ref{balancedDSSP} with some straightforward modifications works also for the DSSP with nearly optimal SO, proposed in Section\,\ref{subA}. The following theorem summarizes  results of Section\,\ref{subC}.  

\begin{theorem}
For any $m$ and $n$, with $k$ being the smallest integer with $m \leq { n\choose k}$, the DSSP with optimal SO together with Algorithm \ref{balancedDSSP}, to construct the access structure $\sA$ and the storing matrix $\boldsymbol{Z}$, has a storage profile vector with bias at most $1$ and a communication complexity vector with bias at most $k$.  

\end{theorem}
\begin{proof}
Suppose that $m = nl +r$, where $0\leq r < n$. Then at the end of Algorithm 1, storage nodes $1,2,\dots,r$ store exactly $l+1$ symbols and storage nodes $r+1,\dots,n$ store exactly $l$ symbols. This is because the index of storing node $i$, that stores data symbol $y_j$ for user $j$, is increased by $1$, modulo $n$, as $j$ increases by $1$ through Algorithm 1.

Now consider an equivalence class $\cF$ of size $d$, as constructed in the proof of \Lref{weakpartition}, and suppose the subsets in $\cF$ are assigned as access sets to users $j+1,\dots,j+d$. Since $\cF$ is a balanced collection, as defined in Definition\,\ref{balanced-def}, and since each user downloads exactly one data symbol from the nodes in its access set, the number of data symbols downloaded by users $j+1,\dots,j+d$ from each of the storage nodes is equal to $\frac{kd}{n}$. At the end of Algorithm 1, the communication complexity vector may not be balanced. This occurs only when the last equivalence class is partially used to assign access sets to users. In that case, the bias of the communication complexity vector is at most $\frac{kd}{n} \leq k$, where $d$ is the size of the last class being used. This completes the proof. 
\end{proof}

Note that in the special case with $m = { n\choose k}$, we have a communication-optimal DSSP with optimal storgae overhead, by \Tref{optimal}, together with balanced communication complexity vector. It also has balanced storage profile vector provided that ${ n\choose k}$ is divisible by $n$, e.g., $n$ being a prime number. As an example, we modify the constructed DSSP in example \ref{optimalDSSPexmple} by using Algorithm \ref{balancedDSSP} to construct $\sA$ and $\boldsymbol{Z}$. 

\begin{exmp}\label{balancedexample}
By modifying the access structure and the storing matrix of the DSSP in example \ref{optimalDSSPexmple} we have
\begin{align*}
&A_1 \hspace{-1mm}=\{1 ,2\},A_2\hspace{-1mm}=\{2 ,3\},A_3\hspace{-1mm}=\{3 ,4\},A_4=\{4 ,5\},A_5\hspace{-1mm}=\{5 ,1\},\\
&A_6\hspace{-1mm}=\{1,3\},A_7\hspace{-1mm}=\{2,4\},A_8\hspace{-1mm}=\{3,5\},A_9\hspace{-1mm}=\{4,1\},A_{10}\hspace{-1mm}=\{5,2\},
\end{align*}
where, by slight abuse of terminology, the first element of $A_j$, in the order written above, stores the data symbol $y_j$. The resulting encoded secrets together with the storage profile is shown in Table \ref{balancedexample}.

\begin{table}[h!]
\centering
 \begin{tabular}{||c c c c c||} 
 \hline
 Node 1 & Node 2 & Node 3 & Node 4 &Node 5 \\ [0.5ex] 
 \hline\hline
 $y_1$ & $y_2$ & $y_3$ & $y_4$ & $y_5$ \\ 
 \hline
 $  2y_3-s_6   $ & $ 2y_4-s_7$ & $2y_5-s_8  $ & $2y_1-s_9  $ &$2y_2-s_{10}$ \\
 \hline
\end{tabular}
\caption{ Storage Profile in Example \ref{balancedexample}}
\label{balancedexample}
\end{table}

\end{exmp}

\section{Conclusion and Future Work}
\label{sec:six}

In this paper, we considered a distributed secret sharing system consisting of a dealer, $n$ storage nodes, and $m$ users. The dealer aims at securely sharing a specific secret $s_j$ with user $j$ via storage nodes, in such a way that no user gets any information about other users' secrets. Given a certain number of storage nodes we find the maximum number of users that can be served in such a system. Also, lower bounds on minimum communication complexity and storage overhead are characterized for any $n$ and $m$. Then we propose distributed secret sharing protocols, under certain conditions on the system parameters, that attain these lower bounds, thereby providing schemes that are optimal in terms of both the communication complexity and storage overhead. Also, the proposed protocols are modified to have balanced communication and storage across the storage nodes.

There are several directions for future work. In this paper, the problems of designing access structure, i.e., which nodes each user has access to, and the coding problem, i.e., how to encode and decode secrets, are considered jointly. In fact, the choices of access structures are not completely arbitrary in our proposed protocols. For different protocols, we discussed sufficient conditions on the access structure. An interesting direction for future work is to study these two problems separately and consider designing efficient coding schemes given a specific access structure. In particular, an interesting problem is the following: what is the necessary and sufficient condition on the access structure that ensures existence of distributed secret sharing protocols with optimal storage overhead and/or minimum communication complexity?

A comparison between the protocol proposed in Section\,\ref{sec:four} providing the perfect secrecy and the one with optimal SO in Section\,\ref{sec:five} satisfying the weak secrecy condition suggests that a fundamental trade-off
exists between the storage overhead and the level of security that a DSSP can offer. In fact, a trade-off between SO and the security level has been recently characterized for the single user case in \cite{chou2020distributed}. In the multi-user case studied in this paper, the two protocols in Section\,\ref{sec:four} and Section\,\ref{sec:five} can be viewed as schemes attaining two extreme points of such a trade-off, one with optimal storage overhead and the other one with the perfect secrecy guarantee. More specifically, a threshold-type secrecy condition can be defined, where the threshold in weak secrecy condition is $1$ and in perfect secrecy condition is $n-1$, the size of $\bs_{-j}$ in \eqref{per-sec-con}. A precise characterization of points that lie between these two points, with both the threshold security parameter and the storage overhead between those of the two extreme points, is an interesting problem and is left for future work.

\section*{Acknowledgment}
We would like to thank Soheil Mohajer for very helpful discussions.

\section*{Appendix A}
\label{appendix}
In this section the solution to continuous optimization problem defined in \eq{minofc}-\eq{spernerc} is determined by satisfying KKT conditions. The Lagrangian can be written as:

\begin{align*}
J=&\sum_{k=1}^{\floor{n/2}} k \alpha_k-\lambda_1 (\sum_{k=1}^{\floor{n/2}}\alpha_k-m)
\\
&-\lambda_2 (1-\sum_{k=1}^{\floor{n/2}}  \frac{\alpha_k}{{{n}\choose{k}}} )- \sum_{k=1}^{\floor{\frac{n}{2}}} \mu_k \alpha_k,
\end{align*}
where $ \lambda_1$,$\lambda_2$ and $\mu_k$'s are Lagrange multipliers. Also KKT conditions are:
\begin{align}
 \forall k: \quad &  k-\lambda_1 + \frac{\lambda_2} {{{n}\choose{k}}} -\mu_k=0\label{kkt1}
 \\
 \forall k:\quad &  \mu_k\geq 0 \label{kkt2}
 \\
 \forall k:\quad & \alpha_k^*\geq 0 \label{kkt3}
 \\
 &\lambda_2 \geq 0  \label{kkt4}
 \\
  \forall k:\quad & \mu_k \alpha_k^*=0\label{kkt5}
 \\
  & \lambda_2 (1-\sum_{k=1}^{\floor{\frac{n}{2}}} \frac{\alpha_k^*}{{{n}\choose{k}}})=0\label{kkt6}
 \\
  & \sum_{k=1}^{\floor{n/2}} \alpha_k^*=m,\label{kkt7} 
\\
 &\sum_{k=1}^{\floor{n/2}}  \frac{\alpha_k^*}{{{n}\choose{k}}} \leq 1.\label{kkt8}
\end{align}
Since both the objective function and inequality constraints are convex and equality condition is an affine function, KKT conditions are sufficient to ensure that the solution is the global minimum. The key point that makes it possible to derive the solution of \eq{kkt1}-\eq{kkt8} is the convexity of the discrete function $f{(k)}\ \deff \ \frac{1}{{n\choose k}}$, assuming $n$ is fixed. Define $m_{(k_1,k_2)}\ \deff \ \frac{f{(k_2)}-f{(k_1)}}{k_2-k_1}$, the slope of the line connecting $(k_1,f{(k_1)})$ and $(k_2,f{(k_2)})$. Also let $d_k \ \deff \ m_{(k,k+1)}$ denote one step increment at point $(k,f_{(k)})$. It is proved in the following lemma that $d_k$ is strictly increasing with $k$.

\begin{lemma}\label{convex}
  $d_k<d_{k+1}$ for all $0\leq k \leq n-2$.
\end{lemma}

\begin{proof}
\begin{align*}
&d_k =\frac{1}{{n\choose k+1}} - \frac{1}{{n\choose k }} =\frac{k!(n-k-1)!(2k-n+1)}{n!}
\\
&\qquad \qquad \qquad \qquad =\frac{(2k-n+1)}{n {{n-1}\choose{k}}}.
\end{align*}
Case $1$: $0 \leq k < \frac{n-1}{2} $. In this case, we have
\begin{align*}
&0\leq(n-1-2(k+1)) < (n-1-2k),
\\
&\quad \qquad 0<\frac{1}{{{n-1}\choose{k+1}}} \leq \frac{1}{{{n-1}\choose{k}}}.
\end{align*}
Multiplying inequalities yields
\begin{align*}
&0\leq \frac{(n-1-2(k+1)) }{{{n-1}\choose{k+1}}} < \frac{(n-1-2k)}{{{n-1}\choose{k}}},
\end{align*}
which implies that
$$
d_k < \ d_{k+1} \leq 0.
$$
Case $2$: $ \frac{n-1}{2}\leq k \leq n-2$. In this case, we have
\begin{align*}
&0\leq (2k-n+1)< (2(k+1)-n+1),
\\
&\quad \qquad 0<\frac{1}{{{n-1}\choose{k}}} < \frac{1}{{{n-1}\choose{k+1}}}.
\end{align*}
Again, multiplying inequalities yields
\begin{align*}
&0\leq  \frac{(2k-n+1)}{{{n-1}\choose{k}}}  <  \frac{(2(k+1)-n+1)}{{{n-1}\choose{k+1}}},
\end{align*}
which implies that
$$
0 \leq d_k < \ d_{k+1}.
$$
This completes the proof of lemma.
\end{proof}

Using \Lref{convex}, an inequality is proved in the following lemma which is used in the proof of \Tref{atmost2}.

\begin{lemma}\label{m-m}
For $k_1,k_2,k_3 \in \N$, with $0\leq k_1 < k_2 < k_3 \leq  {n\choose \floor{\frac{n}{2}}}$, we have:
\begin{equation*}
m_{(k_1,k_2)}<m_{(k_2,k_3)}.
\end{equation*}
\end{lemma}

\begin{proof}
 For $k,k'\in \N$ with $k <k'$,  
 $$
 m_{(k,k')}=\frac{1}{k'-k} \sum_{i=k}^{k'-1} d_i.
 $$
Using this together with \Lref{convex} we have
$$
m_{(k_1,k_2)}=\frac{1}{k_2-k_1} \sum_{i=k_1}^{k_2-1} d_i<\frac{1}{k_2-k_1} \sum_{i=k_1}^{k_2-1} d_{k_2}=d_{k_2},
$$
and
$$
m_{(k_2,k_3)}=\frac{1}{k_3-k_2} \sum_{i=k_2}^{k_3-1} d_{i} \geq \frac{1}{k_3-k_2} \sum_{i=k_2}^{k_3-1} d_{k_2} = d_{k_2},
$$
which conclude the lemma.
\end{proof}



\begin{theorem}\label{atmost2}

In the solution to the optimization problem, at most two of $\alpha_k^*$'s are non-zero. Furthermore, if two of them are non-zero, then their indices are consecutive. 

\end{theorem} 
\begin{proof}\label{atmost2proof}
Assume to the contrary there exist two non-consecutive integers $k_1$ and $k_3$ such that $\alpha_{k_1}^*,\alpha_{k_3}^*\neq0$. Let $k_1<k_3$, without loss of generality. One can find $k_2 \in \N$ such that $k_1<k_2<k_3$. By \eq{kkt5} $\mu_{k_1}$ and $\mu_{k_3}$ must be zero Also, by \eq{kkt1} we can write:
\begin{align*}
k_1-\lambda_1+\lambda_2 \frac{1}{{n\choose k_1}}=0,
\\
k_3-\lambda_1+\lambda_2 \frac{1}{{n\choose k_3}}=0.
\end{align*}
Solving this for $\lambda_1$ and $\lambda_2$ results in
\be{lambda1}
\lambda_1=\frac{\frac{k_3}{{n\choose k_1}}-\frac{k_1}{{n\choose k_3}}}{\frac{1}{{n\choose k_1}}-\frac{1}{{n\choose k_3}}},
\ee

\be{lambda2}
\lambda_2=\frac{k_3- k_1}{\frac{1}{{n\choose k_1}}-\frac{1}{{n\choose k_3}}}.
\ee
By substituting $\lambda_1$ and $\lambda_2$ from \eqref{lambda1} and \eqref{lambda2}, respectively, into (\ref{kkt1}) for $k_2$, 
 $\mu_{k_2}$ is derived as follows: 
\begin{align*}
\mu_{k_2}&=k_2-\lambda_1+\lambda_2 \frac{1}{{n\choose k_2}}=\frac{k_2-k_3}{{n\choose k_1}}+\frac{k_3-k_1}{{n\choose k_2}}+\frac{k_1-k_2}{{n\choose _3}}
\\
&=(k_3-k_2)(\frac{1}{{n\choose k_2}}-\frac{1}{{n\choose k_1}})+(k_2-k_1)(\frac{1}{{n\choose k_2}}-\frac{1}{{n\choose k_3}})
\\
&=(k_3-k_2)(k_2-k_1)(m_{(k_1,k_2)}-m_{(k_2,k_3)}) < 0,
\end{align*}
where the last inequality follows by the assumption on $k_1,k_2,k_3$ and \Lref{m-m}. This contradicts (\ref{kkt2}) which completes the proof.
\end{proof}  
\Tref{atmost2} implies that $\alpha_i^*$,$\alpha_{i+1}^*\geq 0$ for some $i$ and $\alpha_k^*=0$ for $k \neq i,i+1$. Next, $i$, $\alpha_i^*$, and $\alpha_{i+1}^*$ are derived. Note that $\lambda_2 >0$ by \eqref{lambda2} and hence, (\ref{kkt6}) implies that the inequality condition in \eq{kkt8} turns into equality, i.e.,
\be{active}
\frac{\alpha_i^*}{{n\choose i}}+\frac{\alpha_{i+1}^*}{{n\choose i+1}}=1.
\ee
Furthermore, \eqref{kkt7} implies that
\be{totalusers}
\alpha_i^*+\alpha_{i+1}^*=m.
\ee
Therefore, ${\alpha_i^*}$ and $\alpha_{i+1}^*$ can be derived by combining \eqref{active} and \eqref{totalusers} as follows:

\begin{align}\label{alphai}
&\alpha_i^* =\frac{{n\choose i+1} -m}{{n\choose i+1}-{n\choose i}} {n\choose i},
\\
&\alpha_{i+1}^* =\frac{ m-{n\choose i}}{{n\choose i+1}-{n\choose i}} {n\choose i+1}\label{alphai1}.
\end{align}
Note that $\alpha_{i}^*$ and $\alpha_{i+1}^*$ must be non-negative by (\ref{kkt3}). Therefore, $i$ is the largest integer such that
$$
i\leq {n\choose k}.
$$
Also, the minimum of the objective function $\psi$ is given by
\be{min_obj}
\psi^* = i \alpha_i ^*+ (i+1) \alpha_{i+1}^*.
\ee

\section*{Appendix B}
It is shown in the following lemma that the lower bound on $SO$, defined in \eqref{so}, is not tight under the perfect secrecy condition.
\begin{lemma}
\label{lem_so}
If $m>n$, then the storage overhead of a DSSP is strictly greater than $1$ under the perfect secrecy condition.
\end{lemma}
\begin{proof}
It is shown in Section\,\ref{sec:two} that $SO\geq 1$. Hence, it suffices to show that $SO \neq 1$. Assume to the contrary that there exists a perfectly secure DSSP with $SO=1$, i.e., the length of $\by$ is equal to that of $\bs$. This together with noting that  $H(\bs|\by)=0$ by the correctness condition, specified in Definition \ref{ssd}, implies that there is a one-to-one mapping between $\bs$ and $\by$. Furthermore, since $m>n$, there exists a user $j$, for some $j \in [m]$, that has access to at least two data symbols, i.e., the length of $\by_j$ is at least $2$. Then we have
\begin{align}
\label{H}
    H(\bs_{-j}|\by_j)\stackrel{ \text{(a)}}{=}H(\bs|\by_j)\stackrel{ \text{(b)}}{=}H(\by|\by_j)\stackrel{ \text{(c)}}{\leq} (m-2)\log q,
\end{align}
where (a) holds since $s_j$ is a function of $\by_j$, (b) is by noting that there is a one-to-one mapping between $\by$ and $\bs$, and (c) holds since the length of $\by_j$ is at least $2$. 
Moreover, the perfect secrecy condition, specified in \eqref{per-sec-con}, implies that 
\begin{align}
\label{H-perfect}
   H(\bs_{-j}|\by_j)=H(\bs_{-j})=(m-1)\log q,
\end{align}
which holds because the secrets are independent and uniformly distributed. Comparing \eqref{H} with \eqref{H-perfect} shows that the perfect secrecy condition is violated. This contradiction shows that the storage overhead of a perfectly secure DSSP is strictly greater than $1$ for $m>n$.
\end{proof}

\bibliographystyle{IEEEtran}
\bibliography{main}

\begin{thebibliography}{10}
\providecommand{\url}[1]{#1}
\csname url@samestyle\endcsname
\providecommand{\newblock}{\relax}
\providecommand{\bibinfo}[2]{#2}
\providecommand{\BIBentrySTDinterwordspacing}{\spaceskip=0pt\relax}
\providecommand{\BIBentryALTinterwordstretchfactor}{4}
\providecommand{\BIBentryALTinterwordspacing}{\spaceskip=\fontdimen2\font plus
\BIBentryALTinterwordstretchfactor\fontdimen3\font minus
  \fontdimen4\font\relax}
\providecommand{\BIBforeignlanguage}[2]{{%
\expandafter\ifx\csname l@#1\endcsname\relax
\typeout{** WARNING: IEEEtran.bst: No hyphenation pattern has been}%
\typeout{** loaded for the language `#1'. Using the pattern for}%
\typeout{** the default language instead.}%
\else
\language=\csname l@#1\endcsname
\fi
#2}}
\providecommand{\BIBdecl}{\relax}
\BIBdecl

\bibitem{shamir1979share}
A.~Shamir, ``How to share a secret,'' \emph{Communications of the ACM},
  vol.~22, no.~11, pp. 612--613, 1979.

\bibitem{blakley1979safeguarding}
G.~R. Blakley, ``Safeguarding cryptographic keys,'' \emph{Proc. of the National
  Computer Conference1979}, vol.~48, pp. 313--317, 1979.

\bibitem{ben1988completeness}
M.~Ben-Or, S.~Goldwasser, and A.~Wigderson, ``Completeness theorems for
  non-cryptographic fault-tolerant distributed computation,'' in
  \emph{Proceedings of the twentieth annual ACM symposium on Theory of
  computing}.\hskip 1em plus 0.5em minus 0.4em\relax ACM, 1988, pp. 1--10.

\bibitem{chaum1988multiparty}
D.~Chaum, C.~Crepeau, and I.~Damgard, ``Multiparty unconditionally secure
  protocols,'' in \emph{Proceedings of the twentieth annual ACM symposium on
  Theory of computing}.\hskip 1em plus 0.5em minus 0.4em\relax ACM, 1988, pp.
  11--19.

\bibitem{beaver1991foundations}
D.~Beaver, ``Foundations of secure interactive computing,'' in \emph{Annual
  International Cryptology Conference}.\hskip 1em plus 0.5em minus 0.4em\relax
  Springer, 1991, pp. 377--391.

\bibitem{canetti2000security}
R.~Canetti, ``Security and composition of multiparty cryptographic protocols,''
  \emph{Journal of CRYPTOLOGY}, vol.~13, no.~1, pp. 143--202, 2000.

\bibitem{cramer2000general}
R.~Cramer, I.~Damgard, and U.~Maurer, ``General secure multi-party computation
  from any linear secret-sharing scheme,'' in \emph{Advances in Cryptology --
  EUROCRYPT 2000}.\hskip 1em plus 0.5em minus 0.4em\relax Springer, 2000, pp.
  316--334.

\bibitem{garay2000secure}
J.~A. Garay, R.~Gennaro, C.~Jutla, and T.~Rabin, ``Secure distributed storage
  and retrieval,'' \emph{Theoretical Computer Science}, vol. 243, no.~1, pp.
  363--389, 2000.

\bibitem{ateniese2006improved}
G.~Ateniese, K.~Fu, M.~Green, and S.~Hohenberger, ``Improved proxy
  re-encryption schemes with applications to secure distributed storage,''
  \emph{ACM Transactions on Information and System Security (TISSEC)}, vol.~9,
  no.~1, pp. 1--30, 2006.

\bibitem{kumar2012publicly}
P.~S. Kumar, M.~S. Ashok, and R.~Subramanian, ``A publicly verifiable dynamic
  secret sharing protocol for secure and dependable data storage in cloud
  computing,'' \emph{International Journal of Cloud Applications and Computing
  (IJCAC)}, vol.~2, no.~3, pp. 1--25, 2012.

\bibitem{shankar2008alternative}
B.~Shankar, K.~Srinathan, and C.~P. Rangan, ``Alternative protocols for
  generalized oblivious transfer,'' in \emph{International Conference on
  Distributed Computing and Networking}.\hskip 1em plus 0.5em minus 0.4em\relax
  Springer, 2008, pp. 304--309.

\bibitem{tassa2011generalized}
T.~Tassa, ``Generalized oblivious transfer by secret sharing,'' \emph{Designs,
  Codes and Cryptography}, vol.~58, no.~1, pp. 11--21, 2011.

\bibitem{desmedt1992shared}
Y.~Desmedt and Y.~Frankel, ``Shared generation of authenticators and
  signatures,'' in \emph{Advances in Cryptology -- CRYPTO91}.\hskip 1em plus
  0.5em minus 0.4em\relax Springer, 1992, pp. 457--469.

\bibitem{desmedt1993threshold}
Y.~Desmedt, ``Threshold cryptosystems,'' in \emph{Advances in Cryptology --
  AUSCRYPT92}.\hskip 1em plus 0.5em minus 0.4em\relax Springer, 1993, pp.
  1--14.

\bibitem{shoup2000practical}
V.~Shoup, ``Practical threshold signatures,'' in \emph{Advances in Cryptology
  -- EUROCRYPT 2000}.\hskip 1em plus 0.5em minus 0.4em\relax Springer, 2000,
  pp. 207--220.

\bibitem{aliasgari2020private}
M.~Aliasgari, O.~Simeone, and J.~Kliewer, ``Private and secure distributed
  matrix multiplication with flexible communication load,'' \emph{IEEE
  Transactions on Information Forensics and Security}, vol.~15, pp. 2722--2734,
  2020.

\bibitem{aliasgari2019distributed}
------, ``Distributed and private coded matrix computation with flexible
  communication load,'' in \emph{2019 IEEE International Symposium on
  Information Theory (ISIT)}.\hskip 1em plus 0.5em minus 0.4em\relax IEEE,
  2019, pp. 1092--1096.

\bibitem{bhattad2005weakly}
K.~Bhattad, K.~R. Narayanan \emph{et~al.}, ``Weakly secure network coding,''
  \emph{NetCod, Apr}, vol. 104, 2005.

\bibitem{kadhe2014weakly}
S.~Kadhe and A.~Sprintson, ``Weakly secure regenerating codes for distributed
  storage,'' in \emph{2014 International Symposium on Network Coding
  (NetCod)}.\hskip 1em plus 0.5em minus 0.4em\relax IEEE, 2014, pp. 1--6.

\bibitem{huang2016communication}
W.~Huang, M.~Langberg, J.~Kliewer, and J.~Bruck, ``Communication efficient
  secret sharing,'' \emph{IEEE Transactions on Information Theory}, vol.~62,
  no.~12, pp. 7195--7206, 2016.

\bibitem{mceliece1981sharing}
R.~J. McEliece and D.~V. Sarwate, ``On sharing secrets and {Reed-Solomon}
  codes,'' \emph{Communications of the ACM}, vol.~24, no.~9, pp. 583--584,
  1981.

\bibitem{shah2015distributed}
N.~B. Shah, K.~Rashmi, and K.~Ramchandran, ``Distributed secret dissemination
  across a network,'' \emph{IEEE Journal of Selected Topics in Signal
  Processing}, vol.~9, no.~7, pp. 1206--1216, 2015.

\bibitem{bitar2017staircase}
R.~Bitar and S.~El~Rouayheb, ``Staircase codes for secret sharing with optimal
  communication and read overheads,'' \emph{IEEE Transactions on Information
  Theory}, 2017.

\bibitem{sperner1928satz}
E.~Sperner, ``Ein satz {\"u}ber untermengen einer endlichen menge,''
  \emph{Mathematische Zeitschrift}, vol.~27, no.~1, pp. 544--548, 1928.

\bibitem{lubell1966short}
D.~Lubell, ``A short proof of sperner's lemma,'' \emph{Journal of Combinatorial
  Theory}, vol.~1, no.~2, p. 299, 1966.

\bibitem{boyd2004convex}
S.~Boyd and L.~Vandenberghe, \emph{Convex optimization}.\hskip 1em plus 0.5em
  minus 0.4em\relax Cambridge university press, 2004.

\bibitem{geller2004circulant}
D.~Geller, I.~Kra, S.~Popescu, and S.~Simanca, ``On circulant matrices,''
  \emph{Preprint, Stony Brook University}, 2004.

\bibitem{dean2009challenges}
J.~Dean, ``Challenges in building large-scale information retrieval systems,''
  in \emph{Keynote of the 2nd ACM International Conference on Web Search and
  Data Mining (WSDM)}, vol.~10, no. 1498759.1498761, 2009.

\bibitem{melnik2010dremel}
S.~Melnik, A.~Gubarev, J.~J. Long, G.~Romer, S.~Shivakumar, M.~Tolton, and
  T.~Vassilakis, ``Dremel: interactive analysis of web-scale datasets,''
  \emph{Proceedings of the VLDB Endowment}, vol.~3, no. 1-2, pp. 330--339,
  2010.

\bibitem{aktas2019load}
M.~F. Aktas, A.~Behrouzi-Far, E.~Soljanin, and P.~Whiting, ``Load balancing
  performance in distributed storage with regular balanced redundancy,''
  \emph{arXiv preprint arXiv:1910.05791}, 2019.

\bibitem{chou2020distributed}
R.~A. Chou and J.~Kliewer, ``Distributed secure storage: Rate-privacy trade-off
  and {XOR}-based coding scheme,'' \emph{arXiv preprint arXiv:2001.04241},
  2020.

\end{thebibliography}
\vspace{-7cm}
\begin{IEEEbiographynophoto}{Mahdi Soleymani}
(S'18) received his B.S. and M.S.
degrees in Electrical Engineering at Sharif University of Technology, Tehran, Iran, in 2014 and 2016,
respectively. He is currently pursuing his Ph.D. degree in Electrical Engineering and Computer Science
at University of Michigan, Ann Arbor. His research
interests lie in the area of algebraic coding theory with applications to distributed storage systems, wireless networks and distributed computing. 
\end{IEEEbiographynophoto}
\vspace{-7cm}
\begin{IEEEbiographynophoto}{Hessam Mahdavifar}
(S'10, M'12)  is an Assistant Professor in the Department of Electrical Engineering and Computer Science at the University of Michigan Ann Arbor. He received the B.Sc. degree from the Sharif University of Technology, Tehran, Iran, in 2007, and the M.Sc. and the Ph.D. degrees from the University of California San Diego (UCSD), La Jolla, in 2009, and 2012, respectively, all in electrical engineering. He was with the Samsung US R\&D between 2012 and 2016, in San Diego, US, as a staff research engineer. 

He received the NSF career award in 2020. He also received Best Paper Award in 2015 IEEE International Conference on RFID, and the 2013 Samsung Best Paper Award. He also received two Silver Medals at International Mathematical Olympiad in 2002 and 2003, and two Gold Medals at Iran National Mathematical Olympiad in 2001 and 2002. His main area of research is coding and information theory with applications to wireless communications, storage systems, security, and privacy. 
\end{IEEEbiographynophoto}

\end{document}